\documentclass[a4paper,UKenglish,cleveref, autoref, thm-restate]{lipics-v2021}
\pdfoutput=1
\nolinenumbers
\hideLIPIcs

\usepackage{enumerate}
\usepackage[utf8]{inputenc} 
\usepackage[T1]{fontenc}    
\usepackage{amsmath,amsthm,amssymb}
\usepackage[linesnumbered]{algorithm2e}
\usepackage{hyperref}       
\usepackage{url}            
\usepackage{booktabs}       
\usepackage{amsfonts}       
\usepackage{nicefrac}       
\usepackage{microtype}      
\usepackage{lipsum}         
\usepackage{graphicx}

\newcommand{\citep}{\cite}
\newcommand{\citet}{\cite}

\usepackage{doi}
\usepackage{stackengine}
\usepackage{xfrac}
\usepackage{stmaryrd}
\usepackage{multirow}

\usepackage{slashbox}

\swapnumbers

\usepackage{tikz}
\usetikzlibrary{graphs, shapes, arrows.meta, calc, positioning}
\pgfdeclarelayer{bg}    
\pgfsetlayers{bg, main}

\definecolor{earthyellow}{rgb}{0.88, 0.66, 0.37}
\definecolor{bostonuniversityred}{rgb}{0.8, 0.0, 0.0}
\definecolor{persianblue}{rgb}{0.11, 0.22, 0.73}

\def\probcol{\color{persianblue}}

\pgfdeclarelayer{bg}    
\pgfsetlayers{bg, main}

\newcommand{\PLang}[2][]{%
\def\bannachempty{#1}%
\ifx\empty\bannachempty\text{p-}\else\text{p$_{#1}$-}\fi%
\textsc{#2}%
}

\newcommand{\Q}{\mathbb{Q}}
\newcommand{\R}{\mathbb{R}}

\newcommand{\IN}{\mathbb{N}}
\newcommand{\IR}{\mathbb{R}}

\newcommand{\aS}{\mathfrak{S}}

\newtheorem{fact}{Fact}
 
\usepackage{etoolbox}

\makeatletter
\newcommand{\DeclareMathActive}[2]{%
  \expandafter\edef\csname keep@#1@code\endcsname{\mathchar\the\mathcode`#1 }
  \begingroup\lccode`~=`#1\relax
  \lowercase{\endgroup\def~}{#2}%
  \AtBeginDocument{\mathcode`#1="8000 }%
}

\newcommand{\std}[1]{\csname keep@#1@code\endcsname}
\patchcmd{\newmcodes@}{\mathcode`\-\relax}{\std@minuscode\relax}{}{\ddt}
\AtBeginDocument{\edef\std@minuscode{\the\mathcode`-}}
\makeatother

\let\IsInPP=\relax

\DeclareMathActive{=}{\if\IsInPP1{\std{=}}\else\mathrel{\std{=}}\fi}
\newcommand{\compactEquals}[1]{\let\IsInPP=1#1\let\IsInPP=\relax}
\newcommand{\PP}[1]{\mathbb{P}(\compactEquals{#1})}

\newcommand{\pp}[1]{{P}(\compactEquals{#1})}

\newsavebox{\neqbox}
\savebox{\neqbox}{$\neq$}

\newcommand{\existsR}{\ensuremath{\mathsf{\exists\IR}}}
\newcommand{\ETR}{\ensuremath{\textsc{ETR}}}
\newcommand{\succETR}{\ensuremath{\textup{succ}\text{-}\textsc{ETR}}}

\newcommand{\succETRpoly}{\ensuremath{\textup{succ}\text{-}\textsc{ETR}_{\poly}}}

\newcommand{\sumprodETRpoly}{\ensuremath{\Sigma\Pi\text{-}\textsc{ETR}}}
\newcommand{\prodETRpoly}{\ensuremath{\Pi\text{-}\textsc{ETR}}}
\newcommand{\sumETRpoly}{\ensuremath{\Sigma\text{-}\textsc{ETR}}}
\newcommand{\sumviETR}{\ensuremath{\Sigma_{\textit{vi}}\text{-}\textsc{ETR}}}

\newcommand{\QBF}{\ensuremath{\mathsf{QBF}}}
\newcommand{\succR}{\ensuremath{\mathsf{succ\text{-}\exists\IR}}}
\newcommand{\succRpoly}{\ensuremath{\mathsf{succ\text{-}\exists\IR}_{\poly}}}
\newcommand{\NP}{\ensuremath{\mathsf{NP}}}
\newcommand{\NEXP}{\ensuremath{\mathsf{NEXP}}}
\newcommand{\realNP}{\ensuremath{\mathsf{NP}_{\text{real}}}}
\newcommand{\realNEXP}{\ensuremath{\mathsf{NEXP}_{\text{real}}}}
\newcommand{\PSPACE}{\ensuremath{\mathsf{PSPACE}}}
\newcommand{\EXPSPACE}{\ensuremath{\mathsf{EXPSPACE}}}
\newcommand{\VP}{\ensuremath{\mathsf{VP}}}
\newcommand{\VNP}{\ensuremath{\mathsf{VNP}}}

\newcommand{\ccPP}{\ensuremath{\mathsf{PP}}}
\newcommand{\sharpP}{\#\mathsf{P}}

\newcommand{\per}{\ensuremath{\mathsf{per}}}
\newcommand{\sgn}{\operatorname{sgn}}

\newcommand{\Nset}{\mathbb{N}}

\newcommand{\lep}{\leq_{\textit{P}}}
\newcommand{\leqp}{\leq_{\textit{P}}}

\DeclareMathOperator{\poly}{poly}
\DeclareMathOperator{\polylog}{polylog}

\DeclareMathOperator{\NTimereal}{\mathsf{NTime}_{\text{real}}}

\newcommand{\maxvaluecount}{c}
\newcommand{\cL}{{\mathcal L}}
\newcommand{\cE}{{\mathcal E}}

\newcommand{\cT}{{\mathcal T}}

\newcommand{\fA}{{\mathfrak A}}

\newcommand{\fM}{{\mathfrak M}}

\newcommand{\be}{{\bf e}}
\newcommand{\bff}{{\bf f}}

\newcommand{\bp}{{\bf p}}

\newcommand{\bt}{{\bf t}}

\def\probname#1#2{^{\text{#1}}_{\text{#2}}}
\def\probsumname#1#2{^{\textit{#1}{\langle{\scriptscriptstyle\Sigma}\rangle}}_{\textit{#2}}}
\def\probnamesum#1#2{^{\textit{#1}{\langle{\scriptscriptstyle\Sigma}\rangle}}_{\textit{#2}}}

\def\Tpolysum{\cT\probnamesum{\textit{poly}}{}}

\def\Tpoly{\cT\probname{\textit{poly}}{}}

\def\Tlin{\cT\probname{\textit{lin}}{}}
\def\Tcomp{\cT\probname{\textit{base}}{}}

\def\Lstar{\cL^{*}}

\def\Lprobpolysum{\cL\probnamesum{poly}{}}

\def\Lab{\textit{Lab}}

\newcommand{\SAT}{\mbox{\sc Sat}}
\newcommand{\SATprobstar}{\mbox{\sc Sat}^{*}_{{prob}}}

\newcommand{\SATprobcomp}{\mbox{\sc Sat}^{\textit{base}}_{\textit{prob}}}

\newcommand{\SATproblin}{\mbox{\sc Sat}^{\textit{lin}}_{\textit{prob}}}

\newcommand{\SATprobpoly}{\mbox{\sc Sat}^{\textit{poly}}_{\textit{prob}}}

\newcommand{\SATprobpolysum}{\mbox{\sc Sat}\probsumname{poly}{\textit{prob}}}

\newcommand{\SATprobpolysumsm}{\mbox{\sc Sat}\probsumname{poly}{\textit{sm},\textit{prob}}}

\newcommand{\SUM}{\Sigma}

\newcommand{\LESOlabel}{\ensuremath{\mathsf{L}\text{-}\mathsf{ESO}}}
\newcommand{\LESO}[2][]{\LESOlabel\ensuremath{_{[#1]}(#2)}}
\newcommand{\ESOlabel}{\ensuremath{\mathsf{ESO}}}
\newcommand{\LESOforProof}{\LESO[-1,1]{+,\times,-,=,\leq,0,\sfrac{1}{8},1}}

\newcommand{\succETRcR}[2]{\ensuremath{\mathrm{succETR}^{#1,+,\times}_{[#2]}}}

\newcommand{\FOindependent}{\ensuremath{\mathsf{FO}(\indep\negthinspace_c)}}

\newcommand\myldots{\!\makebox[1em][c]{.\hfil.\hfil.}}  
\def\mmid{ | }

\makeatletter
\newcommand*{\indep}{%
  \mathbin{%
    \mathpalette{\@indep}{}%
  }%
}
\newcommand*{\nindep}{%
  \mathbin{
    \mathpalette{\@indep}{/}%
  }%
}
\newcommand*{\@indep}[2]{%
  \sbox0{$#1\perp\m@th$}
  \sbox2{$#1=$}
  \sbox4{$#1\vcenter{}$}
  \rlap{\copy0}
  \dimen@=\dimexpr\ht2-\ht4-.2pt\relax
  \kern\dimen@
  \ifx\\#2\\%
  \else
    \hbox to \wd2{\hss$#1#2\m@th$\hss}%
    \kern-\wd2 %
  \fi
  \kern\dimen@
  \copy0 
}
\title{The Existential Theory of the Reals with Summation Operators}

\ArticleNo{52}
\author{Markus Bl\"{a}ser}{Saarland University, Germany}{mblaeser@cs.uni-saarland.de}{
https://orcid.org/0000-0002-1750-9036}{}
\author{Julian D\"{o}rfler}{Saarland University, Germany}{jdoerfler@cs.uni-saarland.de}{https://orcid.org/0000-0002-0943-8282}{}
\author{Maciej Li\'{s}kiewicz}{University of Lübeck, Germany}{maciej.liskiewicz@uni-luebeck.de}{https://orcid.org/0000-0003-0059-5086}{}
\author{Benito van der Zander}{University of Lübeck, Germany}{b.vanderzander@uni-luebeck.de}{https://orcid.org/0000-0001-5957-4621}{Work supported by the Deutsche Forschungsgemeinschaft (DFG) grant 471183316 (ZA 1244/1-1).}

\authorrunning{M. Bl\"{a}ser, J. D\"{o}rfler, M.  Li\'{s}kiewicz, and B. van der Zander} 

\Copyright{Markus Bl\"{a}ser, Julian D\"{o}rfler, Maciej Li\'{s}kiewicz, and Benito van der Zander} 



\category{} 

\relatedversion{Proceedings version in LIPIcs, Volume 322, ISAAC 2024.} 

\newenvironment{proofidea}{\noindent\emph{Proof idea: }}{\qed}

\begin{document}

\maketitle

\begin{abstract}
To characterize the computational complexity of satisfiability 
problems for probabilistic and causal reasoning within Pearl’s 
Causal Hierarchy, van der Zander, Bl\"aser, and Li{\'s}kiewicz 
[IJCAI 2023] introduce a new natural class, named $\succR$.
This class can be viewed as a succinct variant of the well-studied class
$\existsR$ based on the Existential Theory of the Reals ($\ETR$).
Analogously to $\existsR$, $\succR$ is an intermediate class 
between $\NEXP$ and $\EXPSPACE$, the exponential versions 
of $\NP$ and $\PSPACE$.

The main contributions of this work are threefold.
Firstly, we characterize the class $\succR$ in terms of nondeterministic 
real Random-Access Machines (RAMs) and develop structural complexity theoretic results 
for real RAMs, including translation and hierarchy theorems. Notably, we 
demonstrate the separation of $\existsR$ and $\succR$. Secondly, 
we examine the complexity of model checking and satisfiability of 
fragments of existential second-order logic and probabilistic independence 
logic. We show $\succR$-completeness of several of these problems, for which the 
best-known complexity lower and upper bounds were previously 
$\NEXP$-hardness and $\EXPSPACE$, respectively. Thirdly, while $\succR$ 
is characterized in terms of ordinary (non-succinct) $\ETR$ instances enriched 
by exponential sums and a mechanism to index exponentially many variables, 
in this paper, we prove that when only exponential sums are added, the 
corresponding class $\existsR^\Sigma$ is contained in $\PSPACE$. 
We conjecture that this inclusion is strict, as this class is equivalent 
to adding a $\VNP$-oracle to a polynomial time nondeterministic real RAM. 
Conversely, the addition of exponential products to $\ETR$, yields $\PSPACE$.
Furthermore, we study the satisfiability 
problem for probabilistic reasoning, with the additional requirement 
of a small model, and prove that this problem is complete for $\existsR^\Sigma$.
\end{abstract}

\newpage
\setcounter{page}{1}
       
\section{Introduction}\label{sec:intro}
The existential theory of the reals, $\ETR$, is the set of all true sentences of the form
\begin{equation} \label{eq:etr:1}
   \exists x_0 \dots \exists x_n \ \varphi(x_0,\dots,x_n),
\end{equation}
where $\varphi$ is a quantifier-free Boolean formula over the basis $\{\vee, \wedge, \neg\}$, variables $x_0, \ldots, x_n$,
and a signature consisting of the constants $0$ and $1$, the functional symbols
$+$ and $\cdot$, and the relational symbols $<$, $\le$, and $=$. 
The sentences are interpreted over the real numbers in the standard way.
The significance of this theory lies in its exceptional expressiveness, enabling 
the representation of numerous natural problems across 
computational geometry \cite{abrahamsen2018art,mcdiarmid2013integer,cardinal2015computational}, 
Machine Learning and Artificial Intelligence \cite{abrahamsen2021training,ibeling2022mosse,zander2023ijcai}, 
game theory \cite{bilo2017existential,garg2018r}, 
and various other domains. Consequently, a complexity class, $\exists \mathbb{R}$, 
has been introduced to capture the computational complexity associated with determining 
the truth within the existential theory of the reals. This class is formally defined as the closure 
of $\ETR$ under polynomial-time many-one reductions \citep{DBLP:journals/jsc/GrigorevVV88,existentialTheoryOfRealsCanny1988some,existentialTheoryOfRealsSchaefer2009complexity}.
For a comprehensive compendium on  $\exists \mathbb{R}$, see \cite{schaefer2024existential}.

Our study focuses on $\ETR$ which extends the syntax of formulas to allow the use of summation 
operators in addition to the functional symbols $+$ and $\cdot$. 
This research direction, initiated in \citep{zander2023ijcai}, was motivated by an attempt to accurately
characterize the computational complexity of satisfiability problems for probabilistic and causal reasoning 
across ``Pearl’s Causal Hierarchy''  (PCH) \cite{shpitser2008complete,Pearl2009,bareinboim2022pearl}.

In \citep{zander2023ijcai}, the authors introduce a new natural class, named $\succR$, which can
be viewed as a succinct variant of $\existsR$. Perhaps, one of the most notable complete problems for 
the new class is the problem, called  $\sumviETR$ ("{\it vi}" stands for variable indexing). 
It is defined as an extension of $\ETR$ by adding to the signature 
an additional summation operator\footnote{In \cite{zander2023ijcai}, the authors assume arbitrary integer lower and upper bound in $\sum_{x_j = a}^b$. 
It is easy to see that, w.l.o.g., one can restrict $a$ and $b$ to binary values.} 
$\sum_{x_j = 0}^1$ 
which can be used to \emph{index} the quantified variables $x_i$  used in Formula~\eqref{eq:etr:1}. 
To this end, the authors define
variables of the form $x_{\langle x_{j_1},\dots,x_{j_m}\rangle }$, which represent indexed variables with the index given by $x_{j_1},\dots,x_{j_m}$ interpreted as a number in binary. 
They can only be used when variables $x_{j_1},\dots,x_{j_m}$ occur in the scope of summation operators with range $\{0,1\}$.
E.g., $\exists x_0 \ldots \exists x_{2^N -1}
\sum_{e_1=0}^1 \ldots \sum_{e_N=0}^1 (x_{\langle e_1,\ldots, e_N\rangle})^2=1$
is a $\sumviETR$ sentence\footnote{We represent the instances in $\sumviETR$ 
omitting the (redundant) block of existential quantifiers, so the encoding of the example instance has length polynomial in $N$.} 
encoding a unit vector in $\R^{2^N}$.  Note that sentences of $\sumviETR$ allow 
the use of exponentially many variables.
Another example sentence is 
$\sum_{x_1=0}^1 \sum_{x_2=0}^1 (x_1+ x_2)  (x_1+(1-x_2)) (1-x_1) = 0$
that models the co-{\sc Sat} instance
 $(p\vee q)\wedge(p\vee \overline{q})\wedge\overline{p}$.
It shows that the summation operator can also be used in $\sumviETR$ formulas in a standard way.

Analogously to $\existsR$, $\succR$ is an intermediate class between the exponential versions of $\NP$ and $\PSPACE$: 
\begin{equation}\label{eq:knowns:inclusions}
    \NP \subseteq \existsR \subseteq \PSPACE \subseteq \NEXP \subseteq \succR \subseteq \EXPSPACE.
\end{equation} 
An interesting challenge, in view of the new class, is 
to determine whether it contains harder problems than $\NEXP$ 
and to examine the usefulness of $\succR$-completeness as a tool for 
understanding the apparent intractability of natural problems.
A step in these directions, that we take in this work, is to express $\succR$ 
in terms of machine models over the reals, which in the case of $\existsR$ 
yield an elegant and useful characterization
by $\realNP$  \cite{erickson2022smoothing}.

In our work, we study also the different restrictions on which the summation operators 
in $\ETR$ are allowed to be used and the computational complexity of deciding the resulting 
problems. In particular, we investigate $\existsR^\Sigma$ -- the class based on 
$\ETR$  enriched with standard summation operator, and $\succRpoly$ which is based on 
$\sumviETR$ with the restriction that only polynomially many variables can be used.

In this paper, we 
employ a family of satisfiability problems for probabilistic reasoning, 
which nicely demonstrates the expressiveness of the $\ETR$  variants under consideration 
and illustrates the natural necessity of introducing the summation operator.

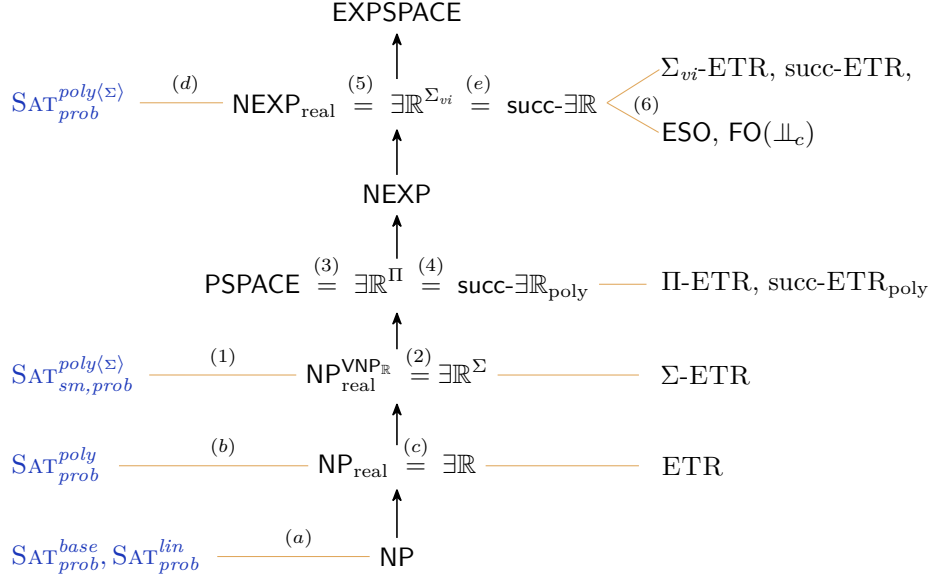
\begin{figure}

\centering

{\begin{tikzpicture}[yscale=0.8,
  inclusion/.style = {
    semithick,
    ->,
    > = {[round,sep]Stealth}
  },
  classbox/.style = {
    fill = blue!25,
    rounded corners,
  },
  contained/.style = {
    thin,
    color=earthyellow 
  }
  ]

  \def\leftX{2}
  \def\leftY{1}
  \def\leftSkip{2}
 
  \def\rightX{3}
  \def\rightY{1}
  \def\rightSkip{1.5}

  \node at (\rightX,\rightY)                           		(np) {$\NP$};
  \node at (\rightX,\rightY+1*\rightSkip)              	(existsR) {$ \realNP~=~\existsR$};
  \node at (\rightX,\rightY+2*\rightSkip)             	(existsRsigma) {$\realNP^{\VNP_\IR}~=\existsR^\Sigma$};
  \node at (\rightX,\rightY+3*\rightSkip)               (pspace) {$\PSPACE~=~\existsR^\Pi~=~\succR_{\rm poly}$};
  \node at (\rightX,\rightY+4*\rightSkip)               (nexp) {$\NEXP$};
  \node at (\rightX,\rightY+5*\rightSkip)               (succR) {$~~~~\realNEXP~=~\existsR^{\Sigma_{\textit{vi}}}~=~\succR$};
  \node at (\rightX,\rightY+6*\rightSkip)               (expspace) {$\EXPSPACE$};

  \graph[use existing nodes, edges = {inclusion}] {
    np -- existsR  -- existsRsigma -- pspace --  nexp -- succR -- expspace;  
  };

  	\def\resultleftX{0}
 	\def\resultrightX{6.5}
  	\def\resultStart{0.75}
  	\def\resultSkip{1.5}


  \node[text width=4.0cm, align=left] at (\rightX+5.5,\rightY+1*\rightSkip)         (ETR){$\ETR$};
  \node[text width=4.0cm, align=left] at (\rightX+5.5,\rightY+2*\rightSkip)         (sumETRpoly){$\sumETRpoly$};
  \node[text width=4.0cm, align=left] at (\rightX+5.5,\rightY+3*\rightSkip)         (succETRpoly){$\prodETRpoly$, $\succETRpoly$};               
  \node[text width=4.0cm, align=left] at (\rightX+5.5,\rightY+5.0*\rightSkip)         (succETR){$\sumviETR$, $\succETR$,\\[4mm]$\ESOlabel$, $\FOindependent$};
  
  \draw[contained, shorten <=2.3cm]    (ETR.center)+(0,0) --   (existsR);  
  \draw[contained, shorten <=2.3cm]    (sumETRpoly.center)+(0,0) --  (existsRsigma);
  \draw[contained, shorten <=2.2cm]    (succETRpoly.center)+(0,0)  --  (pspace);
  
  \draw[contained, shorten <=2.4cm]    (succETR.center)+(0,2.15)  -- (succR.east);
  \draw[contained, shorten <=2.35cm]    (succETR.center)+(0,-1.95)  -- (succR.east); 


   \node[text width=3.2cm, align=left] at (\rightX-3.5,\rightY+0*\rightSkip)        (SATcausallin){\probcol$\SATprobcomp, \SATproblin$};
   \node[text width=3.2cm, align=left] at (\rightX-3.5,\rightY+1*\rightSkip)        (SATcausalpoly){\probcol$\SATprobpoly$};
  \node[text width=3.2cm, align=left] at (\rightX-3.5,\rightY+2*\rightSkip)         (SATsmpolysum){\probcol$\SATprobpolysumsm$};             
  \node[text width=3.2cm, align=left] at (\rightX-3.5,\rightY+5*\rightSkip)         (SATprobpolysum){\probcol$\SATprobpolysum$};
  
   \draw[contained, shorten <=-0.6cm]      (SATcausallin)-- (np);
   \draw[contained, shorten <=-1.9cm]      (SATcausalpoly) -- (existsR);
   \draw[contained, shorten <=-1.5cm]      (SATsmpolysum) -- (existsRsigma);
   \draw[contained]     (\rightX-3.4,\rightY+5*\rightSkip) --  (\rightX-2.3,\rightY+5*\rightSkip) ;


\node at (\rightX-2.3,4.3)                         	(ref5) {\scriptsize$(1)$};
\node at (\rightX-0.48,8.8)                         	(ref4) {\scriptsize$(5)$};
\node at (\rightX+0.45,5.8)                         	(ref3) {\scriptsize$(4)$};
\node at (\rightX-0.93,5.8)                         	(ref2) {\scriptsize$(3)$};
\node at (\rightX+0.29,4.3)                         	(ref1) {\scriptsize$(2)$};
\node at (\rightX+3.3,8.45)                          	(ref6) {\scriptsize$(6)$};

\node at (\rightX-2.8,8.8)                         	(refe) {\scriptsize$(d)$};
\node at (\rightX-2.3,2.8)                         	(refd) {\scriptsize$(b)$};
\node at (\rightX-1.3,1.3)                         	(refc) {\scriptsize$(a)$};
\node at (\rightX+1.08,8.8)                         	(refb) {\scriptsize$(e)$};
\node at (\rightX+0.25,2.8)                    	(refa) {\scriptsize$(c)$};

\end{tikzpicture}}
\caption{The landscape of complexity classes of the existential theories of the reals 
 and the satisfiability problems for the probabilistic languages (to the left-hand side)
complete for the corresponding classes.
Arrows "$\rightarrow$" denote inclusions
$\subseteq$ and the earth-yellow labeled lines "${\color{earthyellow} -}$" connect  
complexity classes with complete problems for those classes. 
The completeness results $(a),(b),$ and $(d)$ were proven by  
Fagin et al.~\cite{fagin1990logic}, 
Moss{\'e} et al.~\cite{ibeling2022mosse}, resp.  
van der Zander  et al.~\cite{zander2023ijcai}.
The characterization $(c)$ is due to Erickson et al.~\cite{erickson2022smoothing}
and $(e)$ is proven in \cite{zander2023ijcai}.
The references to our results are as follows:
$(1)$ Theorem~\ref{thm:main:sec:Sigma:ETR}, 
$(2)$~Theorem~\ref{corollary:sumETRpoly:is:realNP:VNP:IR complete},
$(3)$ and $(4)$~Theorem~\ref{thm:main:pspace:succETRpoly},
$(5)$~Lemma~\ref{lem:realNEXPsuccR}, and
$(6)$~Theorem~\ref{thm:ESO:main}.
}
\label{figure:main:results:summary} 
\end{figure}

%
\paragraph*{Related Work to our Study}
In a pioneering paper in this field, Fagin, Halpern, and Megiddo~\cite{fagin1990logic}
consider the probabilistic language 
consisting of Boolean combinations 
of (in)equalities of \emph{basic} and \emph{linear} terms, like 
$\PP{(X = 0 \vee Y=1) \wedge (X = 0 \vee Y = 0)}\compactEquals{=1} \wedge(\PP{X = 0}\compactEquals{=0}\vee \PP{X = 0}\compactEquals{=1})  \wedge(\PP{Y = 0}\compactEquals{=0}\vee \PP{Y = 0}\compactEquals{=1})$, over binary variables $X,Y$
(which can be seen as a result of reduction from the 
satisfied Boolean formula $(\overline{a} \vee {b})\wedge(\overline{a} \vee \overline{b})$).
The authors provide a complete axiomatization for the used logic
and 
investigate the complexity of 
the probabilistic satisfiability problems $\SATprobcomp$ and $\SATproblin$,  which ask whether 
there is a joint probability distribution of $X,Y,\ldots$  that satisfies
a given Boolean combination of (in)equalities of basic, respectively linear,  terms
(for formal definitions, see Sec.~\ref{sec:preliminaries}).
They show that both satisfiability problems
are  $\NP$-complete (cf.~Fig.~\ref{figure:main:results:summary}). Thus, surprisingly, 
the complexity is no worse than that of propositional logic. Fagin et al.~extend
then the language to (in)equalities of \emph{polynomial} terms, with the goal of 
reasoning about conditional probabilities. They prove that there is a $\PSPACE$ algorithm, based 
on Canny’s decision procedure~\cite{existentialTheoryOfRealsCanny1988some}, for deciding 
if such a formula is satisfiable but left the exact complexity open.
Recently, Moss{\'e}, Ibeling, and Icard, \cite{ibeling2022mosse} have solved this problem, showing that 
deciding the satisfiability ($\SATprobpoly$)
is~$\existsR$-complete.
In \cite{ibeling2022mosse}, the authors also investigate the satisfiability problems for the higher,
more expressive PCH layers -- which are not the subject of our paper -- and prove an interesting result, that 
for (in)equalities of polynomial terms both at the interventional and the counterfactual 
layer the decision problems still remain $\existsR$-complete.

The languages used in~\cite{fagin1990logic,ibeling2022mosse} and also in other relevant works as, e.g.,  
\cite{nilsson1986probabilistic,georgakopoulos1988probabilistic,ibeling2020probabilistic},
can only represent \emph{marginalization} as an expanded sum since they lack a unary~summation~operator~$\Sigma$.
Thus, for instance, to express the marginal distribution of a 
random variable $Y$ over a subset of (binary) variables $\{Z_1,\ldots,Z_m\}$ as
$\sum_{z_1,\ldots,z_m} \PP{y,z_1,\ldots,z_m}$, an encoding without summation requires an extension 
$\PP{y,Z_1=0,\ldots,Z_m=0} + \ldots +\PP{y,Z_1=1,\ldots,Z_m=1}$ of exponential size. 
Thus to analyze the complexity aspects of the standard notation of probability theory, 
one requires an encoding that directly represents marginalization. 
In a recent paper \cite{zander2023ijcai}, the authors introduce the class $\succR$, 
and show that the satisfiability  ($\SATprobpolysum$) for the (in)equalities of \emph{polynomial} terms involving probabilities is  $\succR$-complete.

Thus, $\succR$-completeness seems to be a meaningful yardstick for measuring
computational complexity of decision problems. An interesting task would be to 
investigate problems involving the reals that have been shown to be in $\EXPSPACE$,
but not to be $\EXPSPACE$-complete, which are natural candidates for $\succR$-complete problems.

\paragraph*{Contributions and Structure of the Paper} 
Below we highlight our main  contributions, partially summarized also in Fig.~\ref{figure:main:results:summary}.
\begin{itemize}
\item
	We provide the characterization of $\succETR$ in terms of nondeterministic real RAMs 
	of exponential time respectively (Sec.~\ref{sec:nexp:over:the:reals}).
	Moreover, for the classes over the reals 
	in the sequence of inclusions~\eqref{eq:knowns:inclusions},
	an upward translation result applies, which implies, e.g.,
	 $\NEXP \subsetneq \succR$ unless $\NP = \existsR$ which is widely disbelieved
	 (Sec.~\ref{sec:relationships:bool:class:vs:ove:the:reals}).
\item     
	We strength slightly the completeness result (marked as $(d)$
	in Fig.~\ref{figure:main:results:summary}) of~\cite{zander2023ijcai} 
	and prove the problem $\SATprobpolysum$ remains $\succR$-complete even 
  	if we disallow the basic terms to contain conditional probabilities
	(Sec.~\ref{sec:succR:compl:satsumpoly}).
\item We show that existential second order logic of real numbers is complete for \succR (Sec.~\ref{sec:ESO}).
\item
	$\PSPACE$ has natural characterizations in terms of ETR; 
    It coincides
	both with $\existsR^\Pi$ -- the class based on 
	ETR enriched with standard product operator, and 
    	with $\succRpoly$, defined in terms of the succinct 
	variant of ETR with polynomially many variables (Sec.~\ref{sec:poly:many:variables}).
\item $\existsR^\Sigma$ -- defined similar to $\existsR^\Pi$, but with the addition of a unary summation operator instead -- is contained in $\PSPACE=\existsR^\Pi$. We conjecture that this inclusion is strict, as the class is equivalent 
to $\realNP^{\VNP_\IR}$, machine to be an $\realNP$ model with a $\VNP_\IR$ oracle,
where $\VNP_\IR$ denotes Valiant's $\NP$ over the reals (Sec.~\ref{sec:etr:sum:machine:char}).
\item
 	Unlike the languages devoid of the marginalization operator,
	the crucial small-model property 
	is no longer satisfied. This property says that any satisfiable formula has a model 
	of size bounded polynomially in the input length. Satisfiability with marginalization 
	and with the additional requirement that there is a small model 
	is complete for  $\existsR^{\Sigma}$ at the probabilistic layer (Sec.~\ref{sec:etr:sum:machine:sm}).
\end{itemize}

\section{Preliminaries}\label{sec:preliminaries}
\paragraph*{Complexity Classes Based on the ETR}
The problem $\succETR$ and the corresponding 
class $\succR$ are defined in \cite{zander2023ijcai} as follows.
$\succETR$ is the set of all Boolean circuits $C$ that encode 
a true sentence $\varphi$ as in Equation~$(\ref{eq:etr:1})$ as follows: 
Assume that 
$C$ computes a function $\{0,1\}^N \to \{0,1\}^M$. 
Then $\varphi$ is a tree consisting of $2^N$ nodes, each node being labeled with a symbol of $\{\vee, \wedge, \neg, +, \cdot, <, \le, =\}$, a constant $0$ or $1$, or a variable $x_0,\ldots x_{2^N-1}$.
For the node $i \in \{0,1\}^N$, the circuit computes an encoding $C(i)$ of the description of node $i$, consisting of the label of $i$, its parent, and its two children. 
The tree represents a true sentence, if the value at the root node would become true after applying the operator of each node to the value of its children, whereby the value of constants and variables is given in the obvious way.
As in the case of $\existsR$, to 
$\succR$ belong all languages 
which are polynomial time many-one reducible to $\succETR$.

Besides $\succETR$, \cite{zander2023ijcai} 
introduce more complete problems for $\succR$ as intermediate
problems in the hardness proof. Of particular importance is the problem
$\sumviETR$ that we already discussed in the Introduction. 
Formally, the problem is defined as an extension of $\ETR$ by adding to the signature 
an additional summation operator $\sum_{x_j = 0}^1$ with the following 
semantics\footnote{Recall, in \cite{zander2023ijcai}, the authors assume arbitrary integer lower and upper bound in $\sum_{x_j = a}^b$. 
But it is easy to see that, w.l.o.g., one can restrict $a$ and $b$ to binary values.}:
If an arithmetic term is given by a tree with 
the top gate $\sum_{x_j = 0}^1$ and $t(x_1,\dots,x_n)$ is the term computed 
at the child of the top gate, then the new term computes
$
   \sum_{e = 0}^1 t(x_1,\dots,x_{j-1},e,x_{j+1},\dots,x_n),
$
that is, we replace the variable $x_j$ by a summation variable~$e$, which then runs from $0$ to $1$.
By nesting the summation operator, we are able
to produce a sum with an exponential number of summands.
The main reason why the new summation variables are introduced 
is due to the fact they can be used to \emph{index} the quantified variables $x_i$  used in Formula~\eqref{eq:etr:1}.
Similarly as in $\succETR$,
sentences of $\sumviETR$ allow the use of exponentially many variables, 
however, the formulas are given directly and do not require  
any succinct encoding.

\paragraph*{Probabilistic Languages} 
We always consider discrete distributions 
in the probabilistic languages studied in this paper. We
represent the values  of the random variables as $\mathit{Val} = \{0,1,\myldots, \maxvaluecount - 1\}$ 
and denote by $X_1,X_2, \myldots,X_n$ the random variables used in the input formula.
We assume, w.l.o.g.,~that they all share the same domain  $\mathit{Val}$.
A value of $X_i$ is often denoted by $x_i$ or a natural number.
In this section, we describe syntax and semantics of the probabilistic languages.

By an \emph{atomic} event, we mean an event of the form $X=x$, where 
$X$ is a random variable and $x$ is a value in the domain of $X$. 
The language $\cE$ of propositional formulas over atomic 
 events is 
 the closure of such events under the Boolean operators $\wedge$ and $\neg$:
 $ \bp ::=  X = x  \mid \neg \bp \mid \bp \wedge \bp$.
The probability $\PP{\delta}$ for formulas $\delta \in \cE$ is called \emph{primitive} or \emph{basic term}, from which we build the probabilistic languages.
The expressive power and computational complexity 
of the languages depend 
on the operations applied to the primitives. 
\\
\noindent
\begin{minipage}[t]{0.47\textwidth}
 Allowing gradually more complex operators, we describe 
the languages which are the subject of our studies below. We start with the description of
the languages $\cT^*$ of terms, using the grammars
given to the right.\footnotemark
\end{minipage}\hspace*{2mm}
\begin{minipage}[t]{0.46\textwidth}
\vspace*{-6mm}
\[
\begin{array}{ll}
    \multirow{1.5}{*}{$\Tcomp$}  & \bt~::=\PP{\delta}    \\[1mm]
    \Tlin  & \bt~::=\PP{\delta}    \mid \bt +\bt     \\[1mm]
   \Tpoly & \bt~::=\PP{\delta} \mid \bt+\bt  \mid -\bt \mid \bt \cdot \bt \\[1mm] 
   \Tpolysum &\bt~::= \PP{\delta} \mid \bt + \bt  \mid -\bt \mid \bt \cdot \bt \mid 
 \mbox{$\sum_{x} \bt$} 
 
   \end{array}
\]
\end{minipage}\\[2mm]
\footnotetext{In the given grammars we omit the brackets for readability, but we assume that they can be used in a standard way.}
\indent
In the summation operator $\sum_{x}$, we have 
a dummy variable $x$ 
which ranges over all values $0,1,\myldots, \maxvaluecount - 1$.
The summation $\sum_{x} \bt$ is a purely syntactical 
concept which represents the sum 
$\bt[\sfrac{0}{ x}]  +\bt[\sfrac{1}{x}]+\myldots +\bt[\sfrac{\maxvaluecount - 1}{x}]$,
where by $\bt[\sfrac{v}{x}]$, we mean the expression in which all occurrences of $x$
are replaced with value $v$.
For example,  for  $\mathit{Val} = \{0,1\}$,
the expression 
 $\sum_{x} \PP{Y=1, X=x}$
 semantically represents $\PP{Y=1, X=0} + \PP{Y=1, X=1}$.
%
We note that the dummy variable $x$ is not a (random) variable in the usual sense
and that its scope is defined in the standard way.


The polynomial calculus $\Tpoly$ was originally introduced by Fagin, Halpern, and Megiddo
\citep{fagin1990logic}
to be able to express conditional probabilities by clearing denominators. While this works 
for $\Tpoly$, this does not work in the case of $\Tpolysum$,
since clearing denominators with exponential sums creates expressions that
are too large. But we could introduce basic terms 
of the form $\PP{\delta'\mmid \delta}$ with $\delta,\delta' \in \cE$
explicitly. 
All our hardness proofs work without conditional probabilities
but all our matching upper bounds are still true with explicit 
conditional probabilities.
For example,
expression as
$\PP{X=1} + \PP{Y=2} \cdot \PP{Y=3}$
is  a valid term in $\Tpoly$.

Now, let 
$
\Lab= \{
\text{base}, 
\text{lin}, 
\text{poly}, \text{poly}\langle{\Sigma}\rangle
\}
$
denote the labels of all variants of languages. Then for each   $*\in\Lab$  we define
the languages $\Lstar$ of Boolean combinations of inequalities in a standard way: 
$ \bff ::= \bt \le \bt' \mid \neg \bff \mid \bff \wedge \bff$,
where $ \bt,\bt'$ are terms in  $ {\cal T}^{*} $.

Although the language and its operations may appear rather restricted, all the usual elements of probabilistic  formulas can be encoded. Namely, equality is encoded as greater-or-equal in both directions, e.g. 
$\PP{x} = \PP{y}$ means $\PP{x} \geq \PP{y} \wedge \PP{y} \geq \PP{x}$.
The number~$0$ can be encoded as an inconsistent probability, 
i.e., $\PP{X=1 \wedge X=2}$. 
In a language allowing addition and multiplication, any positive integer can be easily encoded
from the fact $\PP{\top} \equiv 1$, e.g. $4 \equiv (1 + 1) (1 + 1) \equiv (\PP{\top} + \PP{\top}) (\PP{\top} + \PP{\top})$.
If a language does not allow multiplication, one can show that the encoding is still possible.
Note that these encodings barely change the size of the expressions, so allowing or disallowing these additional operators does not affect any complexity results involving these expressions.


We define the semantics of the languages as follows.
Let $\fM=(\{X_1,\myldots, X_n\}, P)$ be a tuple, where 
$P$ is the joint probability distribution of variables $X_1,\myldots, X_n$.
For values $x_1,\myldots,x_n\in \mathit{Val}$ and $\delta\in \cE$,  
we write $x_1,\myldots,x_n \models \delta$
if $\delta$ is satisfied by the assignment $\compactEquals{X_1=x_1,\myldots, X_n=x_n}$.
Denote by $S_{\delta}=\{x_1,\myldots,x_n \mid  x_1,\myldots,x_n \models \delta\}$.
We define $\llbracket \be \rrbracket_{\fM}$, for some expression $\be$,
recursively in a natural way,  starting with basic terms as follows 
$
\mbox{$\llbracket \PP{\delta} \rrbracket_{\fM} = \sum_{x_1,\;\myldots,x_n\in S_{\delta}}\pp{X_1=x_1,\myldots,X_n=x_n}$}
$
and $\llbracket \PP{\delta\mmid \delta'} \rrbracket_{\fM} = \llbracket\PP{\delta \wedge \delta'} \rrbracket_{\fM}/ \llbracket\PP{\delta'} \rrbracket_{\fM}$, assuming that the expression is undefined if  $\llbracket\PP{\delta'} \rrbracket_{\fM}=0$.
For two expressions $\be_1$ and $\be_2$, we define 
$ \fM \models  \be_1 \le  \be_2$,  if and only if, 
$\llbracket \be_1 \rrbracket_{\fM}\le \llbracket \be_2 \rrbracket_{\fM}.$
The semantics for negation and conjunction are defined in the usual way,
giving the semantics for $\fM \models \varphi$ for any 
$\varphi\in \Lstar$.
\paragraph*{Existential Second Order Logic of Real Numbers}
We follow the definitions of \citep{hannula2020logicsLESOandSNP}.
Let $A$ be a non-empty finite set and $\fA=(A,\R,f_1^{\fA}, \myldots, f_r^{\fA},$ $g_1^{\fA}, \myldots, g_t^{\fA})$, 
with $f_i^{\fA}: A^{ar(f_i)}\to \R$ and $g_i^{\fA}\subseteq A^{ar(g_i)}$, be a structure.
Each $g_i^{\fA}$ is an $ar(g_i)$-ary relation on $A$ and each $f_i^{\fA}$ is a weighted real function on $A^{ar(f_i)}$. 
The term $\bt$ is generated by the following grammar:
$
\bt::=  c \mid f(\vec{x}) \mid \bt + \bt \mid \bt - \bt  \mid \bt \times \bt \mid \sum_x \bt ,
$
where $c\in \R$ is a constant (denoting itself),  $f$ is a function symbol, 
and $\vec{x}$ is a tuple of first-order variables.
An assignment $s$ is a total function that assigns a value in $A$ for each first-order variable. 
The numerical value of $\bt$ in a structure $\fA$ under an assignment $s$, 
denoted by $\llbracket \bt \rrbracket_{\fA}^s$, is defined 
recursively in a natural way, starting with $\llbracket  f_i(\vec{x}) \rrbracket_{\fA}^s = f_i^{\fA}(s(\vec{x}))$ and applying the standard rules of real arithmetic.

For operators $O \subseteq \{+,\times,\SUM,-\}$, (in-)equality operators $E \subseteq \{\leq,<,=\}$, 
and constants $C\subseteq \R$, 
the grammar of  $\ESOlabel_{\R}(O,E,C)$ sentences is given by 
$\phi ::= x=y\mid \neg (x=y) \mid i\ e\ j \mid \neg (i\ e\ j) 
\mid R(\vec{x}) \mid \neg R(\vec{x}) 
\mid \phi \wedge \phi \mid \phi \vee \phi \mid \exists x\;\phi \mid \forall x\;\phi \mid \exists f\;\phi$,
where $x,y\in A$ are first order variables, 
$i, j$ are real terms constructed using operations from $O$
and constants from $C$,  $e \in E$, and $R$ denotes a relation symbol of 
a finite relational vocabulary\footnote{The grammar of  \citep{hannula2020logicsLESOandSNP} does not allow quantification over relations, e.g. $\exists R$, as these relations can be replaced by functions, e.g. chosen by $\exists f$.} $g_1,\ldots,g_t$. 

The semantics of $\ESOlabel_{\R}(O,E,C)$ is defined via $\R$-structures
and assignments analogous to first-order logic with additional 
semantics for second order existential quantifier $\exists f$.
%
That is, a structure $\fA$ satisfies a sentence $\phi$ under an assignment $s$, i.e., $\fA \models_s \phi$, according to the following cases of the grammar:
$\fA \models_s x = y$, iff $s(x)$ equals $s(y)$; 
$\fA \models_s \neg (\phi)$ iff $\fA \not\models_s \phi$; 
$\fA \models_s i\ e\ j$ iff $\llbracket i \rrbracket_{\fA}^s\ e\ \llbracket j \rrbracket_{\fA}^s$ where $\llbracket i \rrbracket_{\fA}^s$ is the numerical value of $i$ as defined above; 
$\fA \models_s R(\vec{x})$ iff $g_i^{\fA}(s(\vec{x}))$ is true for the $g^\fA_i$ corresponding to $R$ in the model $\fA$;
$\fA \models_s \phi \wedge \phi'$ iff $\fA \models_s \phi$ and $\fA \models_s \phi'$;
$\fA \models_s \phi \vee \phi'$ iff $\fA \models_s \phi$ or $\fA \models_s \phi'$;
$\fA \models_s \exists x \phi$ iff $\fA \models_{s[a/x]} \phi$ for some $a \in A$ where $s[a/x]$ means the assignment $s$ modified to assign $a$ to $x$;
$\fA \models_s \forall x \phi$ iff $\fA \models_{s[a/x]} \phi$ for all $a \in A$; and 
$\fA \models_s \exists f \phi$ iff $\fA[h/f] \models_s \phi$ for some\footnote{Note that $h$ might be an arbitrary function and is not restricted to the functions $f^\fA_i$ of the model.} function $h: A^{ar(f)} \to \R$ where $\fA[h/f]$ is the expansion of $\fA$ that interprets $f$ as $h$.

For a set $S\subseteq \R$, 
we consider the restricted logic $\ESOlabel_S(O,E,C)$ and $\LESOlabel_S(O,E,C)$.
There only the operators and constants of $O\cup E\cup C$ are allowed and 
all functions $f$ are maps into $S$, i.e. $f: A^{ar(f)} \to S$.
In the loose fragment $\LESOlabel_S(O,E,C)$, negations $\neg (i\ e\ j)$ on real terms are also disallowed.

Probabilistic independence logic $\FOindependent$ is defined as the
extension of first-order logic with probabilistic independence
atoms $\vec{x}\indep_{\vec{z}}\vec{y}$ whose semantics is the standard semantics of
conditional independence in probability distributions \citep{durand2018probabilistic,hannula2020logicsLESOandSNP}.

\paragraph*{Known Completeness and Complexity Results}
\label{sec:sat:problems}
The decision problems $\SATprobstar$, with $*\in \Lab$,
take as input a formula $\varphi$ 
in  the 
languages $\Lstar$
and  ask whether there exists a model 
$\fM$ such that $\fM \models\varphi$.
%
%
The computational complexity of probabilistic  satisfiability problems has 
been a subject of intensive studies for languages which  do not allow explicitly marginalization 
via summation operator $\Sigma$. Very recently \cite{zander2023ijcai}
addressed the problem for polynomial languages. 

Below, we summarize these results\footnote{In 
the papers \citep{ibeling2022mosse} and \citep{zander2023ijcai} the authors show even stronger results,
namely that the completeness results also hold for causal satisfiability problems.}, 
informally presented in the Introduction:
\begin{itemize} 
\item $\SATprobcomp$   and $\SATproblin$ are $\NP$-complete,
   \citep{fagin1990logic},
\item  $\SATprobpoly$ is  $\existsR$-complete  \citep{ibeling2022mosse}, and
\item $\SATprobpolysum$ is  $\succR$-complete \citep{zander2023ijcai}.
\end{itemize}

For a logic $L$, the satisfiability problem $\SAT(L)$ is defined as follows: 
given a formula $\varphi \in L$, decide whether $\varphi$ is satisfiable. 
For the model checking problem of a logic $L$, we consider the following variant: 
given a sentence $\varphi \in L$ and a structure $\fA$, decide whether 
$\fA \models \varphi$.
For model checking of $\FOindependent$, the best-known complexity lower and upper bounds 
are $\NEXP$-hardness and $\EXPSPACE$, respectively \cite{hannula2023logics}.

\section{\texorpdfstring{$\NEXP$}{NEXP} over the Reals}\label{sec:nexp:over:the:reals}

In \cite{erickson2022smoothing}, Erickson, van Der~Hoog, and  Miltzow
extend the definition of word RAMs to real computations.
In contrast to the so-called BSS model of real computation \citep{blum1989theory}, the real RAMs
of Erickson et al.\ provide integer and real computations at the same time,
allowing for instance indirect memory access to the real registers and other features
that are important to implement algorithms over the reals.
The input to a real RAM is a pair of vectors, the first one is a vector
of real numbers, the second is a vector of integers. Real RAMs have two types
of registers, word registers and real registers. The word registers can store
integers with $w$ bits, where $w$ is the word size. The total number of 
registers is $2^w$ for each of the two types. Real RAMs perform arithmetic operations
on the word registers, 
where words are interpreted as integers between $0$ and $2^w - 1$,
and bitwise Boolean operations. On the real registers, only arithmetic
operations are allowed. Word registers can be used for indirect addressing on both types of registers and the control flow is established by conditional jumps
that depend on the result of a comparison of two word registers or of a real register
with the constant $0$. For further details we refer to the original paper 
\cite{erickson2022smoothing}.

The real RAMs of \cite{erickson2022smoothing} characterize the existential theory of the reals.
The authors prove that a problem is in $\existsR$ iff there is a polynomial
time real verification algorithm for it. 
In this way, real RAMs are an ``easy to program'' mechanism to prove that
a problem is contained in $\existsR$. Beside the input $I$, which is a sequence of words,
the real verification algorithm also gets a certificate consisting of a sequence of real numbers $x$ and a further sequence of words $z$. $I$ is in the language if there is a 
pair $(x,z)$ that makes the real verification algorithm accept. $I$ is not in the language
if for all pairs $(x,z)$, the real verification rejects.

Instead of using certificates and verifiers, we can also define 
nondeterministic real RAMs that can guess words and real numbers on the fly.
Like for classical Turing machines, it is easy to see that these two definitions
are equivalent (when dealing with time bounded computations). 

\begin{definition}
    \label{def:ntimereal}
    Let $t: \IN \to \IN$ be a function.
    We define $\NTimereal(t)$ to be the set of all languages $L \subseteq \{0, 1\}^\star$, such that there is a constant $c \in \IN$ and a nondeterministic real word-RAM $M$ that recognizes $L$ in time $t$ for all word-sizes $w \geq c \cdot \log(t(n)) + c$.

    For any set of functions $T$, we define $\NTimereal(T) = \bigcup_{t \in T} \NTimereal(t)$.
    We define our two main classes of interest, $\realNP$ and $\realNEXP$ as follows:
    \begin{align*}
        \realNP &= \NTimereal(\poly(n)), & \realNEXP &= \NTimereal(2^{\poly(n)}).
    \end{align*}
\end{definition}

Note that the word size needs to be at least logarithmic in the running time,
to be able to address a new register in each step.

One of the main results of Erickson et al. (Theorem 2 in their paper) 
can be rephrased as $\existsR = \realNP$. Their techniques can be extended
to prove that $\succR =\realNEXP$.


We get the following in analogy to the well-known results that
the succinct version of $3\text{-}\textsc{Sat}$ is $\NEXP$-complete.
\begin{lemma} \label{lem:realNEXPsuccR}
    $\succETR$ is $\realNEXP$-complete and thus $\realNEXP = \succR$.
\end{lemma}

\begin{proofidea}
For the one direction, one carefully has to analyze the construction
by Erickson et al.\ and show that the simulation there can also be implemented
succinctly. The reverse direction simply follows from expanding the succinct ETR instance
and use the fact that nondeterministic real word-RAMs can solve ETR in polynomial time.
Along the way, we also obtain a useful normalization procedure for succinct ETR
instances (Lemma~\ref{lem:removeneg}). While for normal ETR instances,
it is obvious that one can always push negations down, it is not clear for
succinct instances. We achieve this in Lemma~\ref{lem:removeneg}.
\end{proofidea}
  
\section{The Relationships between the Boolean Classes and Classes over the Reals}
\label{sec:relationships:bool:class:vs:ove:the:reals}
Now 
we study the new class 
$\realNEXP = \succR$ from a complexity theoretic point
of view. 
\begin{equation} \label{eq:translate:1}
  \NP \subseteq \existsR \subseteq \PSPACE;\hspace{2cm}\NEXP \subseteq \succR \subseteq \EXPSPACE.
\end{equation}
The left side of (\ref{eq:translate:1}) is well-known.
The first inclusion on the right side is obvious, since a real RAM simply can
ignore the real part. The second inclusion follows from 
expanding the succinct instance into an explicit formula
(which now has exponential size) and simply using the known
$\PSPACE$-algorithm.

We prove 
two translation results, that is,
equality of one of the inclusion in the left equation of  $(\ref{eq:translate:1})$
implies the equality of the corresponding inclusion
in the right equation of $(\ref{eq:translate:1})$.

\begin{theorem}\label{thm:translation:lower}
    If $\existsR = \NP$, then $\succR = \NEXP$.
\end{theorem}
\begin{theorem}\label{thm:translation:upper}
    If $\existsR = \PSPACE$, then $\succR = \EXPSPACE$.
\end{theorem}

Further we prove a nondeterministic time hierarchy theorem (see Lemma~\ref{lemma:proper:incl} for the details) for real word RAMs.
Using the characterization of $\existsR$ and $\succR$ in terms of real word RAMs, in particular, we get that $\succETR$ is strictly more expressive than $\ETR$.

\begin{corollary}
    \label{cor:existsR_hierarchy}
    $\existsR = \realNP \subsetneq \realNEXP = \succR$.
\end{corollary}

\section{Hardness of Probabilistic Satisfiability without Conditioning}
\label{sec:succR:compl:satsumpoly}
To prove that $\SATprobpolysum$ is $\succR$-complete, van der Zander,  
Bl\"{a}ser and Li\'{s}kiewicz \citep{zander2023ijcai} show the hardness part 
for the variant of the probabilistic language where the primitives are also allowed 
to be conditional probabilities. 
A novel contribution of our work is to extend this completeness result
to our version for languages which disallow conditional probabilities:

\begin{theorem}\label{thm:succR:completeness:of:L3}
    The problem $\SATprobpolysum$ remains $\succR$-complete even without conditional probabilities.
\end{theorem}
In the rest of this section, we will give the proof of the theorem.

In \citep{zander2023ijcai} the authors
have already shown that $\sumviETR$ is $\succR$-complete.
We define $\sumviETR_1$ in the same way as $\sumviETR$, but asking 
the question whether there is a solution where the sum of the absolute 
values ($\ell_1$ norm) is bounded by $1$.
Then we can reduce $\sumviETR_1$ to $\SATprobpolysum$ 
without the need for conditional probabilities (Lemma~\ref{lem:satprobpolysum_succR_hard}).
The proof that $\sumviETR_1$ is hard for $\succR$ (Lemma~\ref{lem:sumvietr1:reduction})
depends on a result of Grigoriev and Vorobjov~\cite{DBLP:journals/jsc/GrigorevVV88} 
who showed that the solution to an ETR instance can be bounded by a constant that 
only depends on the bitsize of the instance. Thus the solution can be scaled to fit into a probability distribution. 
This completes the proof of Theorem~\ref{thm:succR:completeness:of:L3}.

\begin{theorem}[Grigoriev and Vorobjov \cite{DBLP:journals/jsc/GrigorevVV88} ] \label{thm:grigoriev}
Let $f_1,\dots,f_k \in \IR[X_1,\dots,X_n]$
be polynomials of total degree $\le d$
with coefficients of bit size $\le L$. Then every connected 
component of $\{x \in \IR^n \mid f_1(x) \ge 0 \wedge \dots
\wedge f_k(x) \ge 0 \}$ contains a point of distance
less than $2^{L d^{cn}}$ from the origin 
for some absolute constant $c$. The same is true if some
of the inequalities are replaced by strict inequalities.
\end{theorem}

\begin{lemma}\label{lem:sumvietr1:reduction}
    $\sumviETR \leqp \sumviETR_1$.
\end{lemma}
\begin{proof}
    Let $\phi$ be an instance of $\sumviETR$.
    We will transform it into a formula $\varphi$ such that $\varphi$ has a solution with $\ell_1$ norm bounded by $1$ iff $\phi$ has any solution.

    Let $S$ be the bit length of $\phi$.
    The number $n$ of variables in $\phi$ is bounded by $2^S$.
    The degree of all polynomials is bounded by $S$.
    Note that the exponential sums do not increase the degree at all.
    Finally, all coefficients have bit size $O(S)$.
    Note that one summation operator doubles the coefficients at most.
    By Theorem~\ref{thm:grigoriev}, if $\phi$ is satisfiable, then there is a solution with entries bounded by $T := 2^{2^{2^{cS}}}$ for some constant $c$.

    In our new instance $\varphi$ first creates a small constant $d \leq 1/((2^m+n)T)$ for some $m$ polynomial in $S$ defined below.
    This can be done using Tseitin's trick: We take $2^m$ many fresh variables $t_i$ and start with $(2^m+2^S)t_1 = 1$ and then iterate by adding the equation $\sum_{i = 1}^{2^m-1} (t_i^2 - t_{i+1})^2 = 0$, i.e.\ forcing $t_{i+1} = t_i^2$.
    To implement the first equation we replace $2^m$ by $\sum_{e_1=0}^1 \cdots \sum_{e_m=0}^1 1$ and similarly replace $2^S$.
    To implement the second equation we replace $\sum_{i = 1}^{2^m-1} (t_i^2 - t_{i+1})^2$ by $\sum_{e_1=0}^1 \cdots \sum_{e_m=0}^1\sum_{f_1=0}^1 \cdots \sum_{f_m=0}^1 (t_{e_1, \ldots, e_m}^2 - t_{f_1, \ldots, f_m})^2 \cdot A(e_1, \ldots, e_m, f_1, \ldots, f_m)$ where $A$ is an arithmetic formula returning $1$ iff the binary number represented by $f_1, \ldots, f_m$ is the successor of the binary number represented by $e_1, \ldots, e_m$.
    The number $m$ is polynomial in $S$.
    The unique satisfying assignment to the $t_i$ has its entries bounded by $1/(2^m+2^S)$. 
    Let $d := t_{2^m}$ be the last variable.

    Now in $\phi$ we replace every occurrence of $x_i$ by $\sfrac{x_i}{d}$ 
    and then multiple each (in-)equality by an appropriate power of $d$ to remove the divisions in order to obtain $\varphi$.
    In this way, from every solution to $\phi$, we obtain a solution to $\varphi$ by multiplying the entries by $d$ and vice versa.
    Whenever $\phi$ has a solution, then it has one with entries bounded by $T$.
    By construction $\varphi$ then has a solution with entries bounded by $1/(2^m+2^S)$.
    Since each entry of the solution is bounded by $1/(2^m + 2^S)$, the $\ell_1$ norm is bounded by $1/2$.
\end{proof}

\begin{lemma}\label{lem:satprobpolysum_succR_hard}
    $\sumviETR_1 \leqp \SATprobpolysum$ via a reduction without the need for conditional probabilities.
\end{lemma}
\begin{proof}
    Let $X_0$ be a random variable with range $\{-1, 0, 1\}$ and let $X_1, \ldots, X_N$ be binary random variables.
    We replace each real variable $x_{e_1, \ldots, e_N}$ in the $\sumviETR_{1}$ formula as follows:
    \begin{align*}
        x_{e_1, \ldots, e_N} := \PP{X_0 = 1 \land X_1=e_1 \land \ldots \land X_N = e_N} - \PP{X_0 = -1 \land X_1=e_1 \land \ldots \land X_N = e_N}
    \end{align*}
    This guarantees that $x_{e_1, \ldots, e_N} \in [-1, 1]$.
    The existential quantifiers now directly correspond to the existence of a probability distribution $P(X_0, \ldots, X_N)$, where each variable corresponds to an different set of two entries of $P$.

    Let $P(X_0, \ldots, X_N)$ be a solution to the constructed $\SATprobpolysum$ instance.
    Then clearly setting $x_{e_1, \ldots, e_N} = P(1, e_1, \ldots, e_N) - P(-1, e_1, \ldots, e_N)$ satisfies the original $\sumviETR_1$ instance.
    Furthermore it has an $\ell_1$ norm bounded by $1$:
    \begin{align*}
        \sum_{e_1=0}^1\cdots\sum_{e_N=0}^1|x_{e_1, \ldots, e_N}| &= \sum_{e_1=0}^1\cdots\sum_{e_N=0}^1|P(1, e_1, \ldots, e_N) - P(-1, e_1, \ldots, e_N)| \\
        &\leq \sum_{e_1=0}^1\cdots\sum_{e_N=0}^1(P(1, e_1, \ldots, e_N) + P(-1, e_1, \ldots, e_N))\\
        &\leq 1\,.
    \end{align*}
    
    Vice-versa, let the original $\sumviETR_1$ be satisfied by some choice of the $x_{e_1, \ldots, e_N}$ with $\ell_1$ norm $\alpha$ bounded by $1$.
    We define the probability distribution
    \begin{align*}
        P(X_0, X_1, \ldots, X_N) = \begin{cases}
            \frac{1 - \alpha}{2^N} & \text{if $X_0 = 0$}\\
            \max(x_{X_1, \ldots, X_N}, 0) & \text{if $X_0 = 1$}\\
            \max(-x_{X_1, \ldots, X_N}, 0) & \text{if $X_0 = -1$}\\
        \end{cases}
    \end{align*}
    Every entry of $P$ is non-negative since $\alpha \leq 1$.
    Furthermore the sum of all entries is exactly $1$, the entries with $X_0 \in \{-1, 1\}$ contribute exactly $\alpha$ total and the $2^N$ entries with $X_0 = 0$ contribute $1 - \alpha$ total.
    Since $P$ fulfills the equation $x_{e_1, \ldots, e_N} = P(1, e_1, \ldots, e_N) - P(-1, e_1, \ldots, e_N)$, it is a solution to the constructed $\SATprobpolysum$ instance.
\end{proof}


\section[Correspondence to Existential Second Order Logic and First Order Logic of Probabilistic Independence]{Correspondence to Existential Second Order Logic and \FOindependent}
\label{sec:ESO}
In this section, we investigate the complexity of existential second order logics and the probabilistic independence logic \FOindependent.

%
\begin{lemma}\label{lem:ESO:in:succETR}
Model checking of $\ESOlabel_{\R}(\SUM,+,\times,\leq,<,=,\Q)$ is in \succR.
\end{lemma}
\begin{proof}
In model checking, the input is a finite structure $\fA$ and a sentence $\phi$, and we need to decide whether $\fA \models \phi$. $\fA$ includes a domain $A$ for the existential/universal quantifiers over variables.
Any function (relation) of arity $k$ can be represented as a (Boolean) table of size $|A|^k$.
Some of these tables might be given in the input. The remaining tables of functions chosen by quantifiers $\exists f$ can simply be guessed by  a $\realNEXP$ machine in non-deterministic exponential time.
Then all possible values for the quantifiers of the finite domain can be enumerated and all sentences can be evaluated. This completes the proof, due to the characterization given in Lemma~\ref{lem:realNEXPsuccR}.
\end{proof}

\begin{proposition}\label{prop:MC:ESO:succETR:hard}
    Model checking of \LESO[0,1]{+,\times,\leq,0,1} is \succR-hard.
\end{proposition}
\begin{proof}
We start with the following equivalences relating the logics:
$$ \LESO[0,1]{+,\times,\leq,0,1} \equiv \LESO[-1,1]{+,\times,\leq,0,1}  \equiv   \LESOforProof . $$ 
The first equivalence has been shown by Hannula et al.~\citet{hannula2020logicsLESOandSNP}.
To see the second one, note that we can replace operator $=$ using $a = b$ as $a\leq b \wedge b \leq a$.
The negative one $-1$ can be defined by a function $-1$ of arity $0$ using
$\exists (-1): (-1) + 1 = 0$. 
Then any subtraction $a - b$ can be replaced with $a + (-1) \times b$.
Finally, the fraction $\sfrac{1}{8}$ is a function given by $\exists \sfrac{1}{8}: \sfrac{1}{8}+\sfrac{1}{8}+\sfrac{1}{8}+\sfrac{1}{8}+\sfrac{1}{8}+\sfrac{1}{8}+\sfrac{1}{8}+\sfrac{1}{8} = 1$. These equivalence reductions can be performed in polynomial time.

In the rest of the proof, we show the hardness, reducing the problems in \succR\ to the existential second order logic $\LESOforProof$. 
To this aim, we use a \succR-complete problem 
which is based on a problem given by Abrahamsen et al.~\citep{abrahamsen2018art}, who have shown that an equation system consisting  of only sentences of the form $x_i = \sfrac{1}{8}$, $x_i+x_j=x_k$, and $x_i \cdot x_j = x_k$ is \existsR-complete. 
As shown in  \citet{zander2023ijcai}, this can be turned into a \succR-complete problem, denoted as $\succETRcR{\sfrac{1}{8}}{-\sfrac{1}{8},{\sfrac{1}{8}}}$, by replacing the explicit indices $i,j,k$ with circuits that compute the indices for an exponential number of these three equations. The circuits can be encoded with arithmetic operators, which allows us to encode all equations in  existential second order logic in a  polynomial time reduction.

%
%
%
%
%
%
The instances of $\succETRcR{\sfrac{1}{8}}{-\sfrac{1}{8},{\sfrac{1}{8}}}$ are 
represented as seven Boolean circuits 
$C_0,C_1,\myldots,C_6: \{0,1\}^M \to \{0,1\}^N$
such that 
$C_0(j)$ gives the index of the variable in the $j$th equation of type $x_i=\sfrac{1}{8}$,
$C_1(j),C_2(j),C_3(j)$ give the indices of variables in the $j$th equation of the type $x_{i_1}+x_{i_2}= x_{i_3}$, and 
$C_4(j),C_5(j),C_6(j)$ give the indices of variables in the $j$th equation of the type $x_{i_1} x_{i_2}= x_{i_3}$.
Without loss of generality, we can assume that each type has the same number $2^M$ of equations. 
An instance of the problem $\succETRcR{\sfrac{1}{8}}{-\sfrac{1}{8},{\sfrac{1}{8}}} $ is satisfiable if and only if: 
\begin{align}\label{eq:sat:cond:etr}
&\exists x_0,\myldots, x_{2^N-1} \in [-\sfrac{1}{8},\sfrac{1}{8}]: \forall j \in [0, 2^M-1]:\nonumber\\
    & \quad \quad \quad x_{C_0(j)} = \sfrac{1}{8}, \ \ 
     x_{C_1(j)} + x_{C_2(j)} = x_{C_3(j)},   \ \ \mbox{and} \ \
     x_{C_4(j)} \cdot  x_{C_5(j)} = x_{C_6(j)}.
\end{align}
Below, we prove that 
\begin{equation}\label{prop:MC:FO:succETR:hard}
  \succETRcR{\sfrac{1}{8}}{-\sfrac{1}{8},{\sfrac{1}{8}}} \lep
  \LESOforProof  .
\end{equation}

Let the instance of $\succETRcR{\sfrac{1}{8}}{-\sfrac{1}{8},{\sfrac{1}{8}}} $ be
represented by seven Boolean circuits $C_0,C_1,\myldots,C_6: \{0,1\}^M \to \{0,1\}^N$ as described above.
Let  the variables of the instance 
be indexed as $x_{e_1,\ldots,e_N}$, with $e_i\in \{0,1\}$ for $i\in[N]$. 
We will identify the bit sequence $\vec{b} = b_1,\myldots, b_M$ by an integer $j$, with $0\le j \le 2^M -1$,
the binary representation of which is $b_1\myldots b_M$ and vice versa. 

We construct sentences in the logic $\LESOforProof$ and prove that 
  a binary model satisfies the sentences if and only if the formula~\eqref{eq:sat:cond:etr}
is satisfiable.


Let $q$ be an $N$-ary function where $q(e_1,\ldots,e_N)$ should encode the value of variable $x_{e_1,\ldots,e_N}$. 
For the range, we require $\forall \vec{x}: 0 - \sfrac{1}{8} \leq q(\vec{x})  \wedge  q(\vec{x}) \leq \sfrac{1}{8} $. 

For each circuit $C_i$, we define a function $y_{i}$ whose value $y_{i}(\vec{b})$ is $x_{C_i(j)}$, i.e.,  $q(C_i(j))$. Then $y_i$ can directly be inserted in the equation system~(\ref{eq:sat:cond:etr}). For this, we need to encode the circuit as logical sentences and relate $y$ and $q$.

To model a Boolean formula encoded by a node of $C_i$, with $i=0,1,\myldots, 6$, we use one step of arithmetization  to go from logical formulas to calculations on real numbers, where $0\in\R$ means false and $1 \in \R$ means true. While negation is not allowed directly in \LESOlabel, on the real numbers we can simulate negation by subtraction.

For each node $v$ of each circuit $C_i$, we need a function $c_{i,v}$ of arity $M$, such that
$c_{i,v}(\vec{b})$  
is the value computed by the node if the circuit is evaluated on input $j=b_1\myldots b_M$.

If $v$ is an input node, the node only reads one bit $u_{i,k}$ from the input, so let 
$\forall\vec{b}:c_{i,v}(\vec{b}) = id(b_{k}) $, where $id$ is a function that maps ${0,1}$ from the finite domain to ${0,1} \in \R$.
 
For each internal node $v$ of $C_i$, we proceed as follows.

If $v$ is labeled with $\neg$ and $u$ is a child of $v$, then we require $\forall \vec{b}: c_{i,v}(\vec{b}) = 1 - c_{i,u}(\vec{b})$. 

If $v$ is labeled with $\wedge$ and $u$ and $w$ are children  
of $v$, then we require $\forall \vec{b}: c_{i,v}(\vec{b}) = c_{i,u}(\vec{b}) \times c_{i,w}(\vec{b})$. 

Finally, if $v$ is labeled with $\vee$ and $u$ and $w$ are children  
of $v$, then we require $\forall \vec{b}: c_{i,v}(\vec{b}) = 1 - (1-c_{i,u}(\vec{b}) )\times (1- c_{i,w}(\vec{b}))$.

Thus, if $v$ is an output node of a circuit $C_i$, then, 
for  $C_i$ fed with input $j=b_{1}\myldots b_{M}\in \{0,1\}^M$, we have that $v$ evaluates to true if and only if
$c_{i,v}(\vec{b})=1$. 
 
%
Next, we need an $(N+M)$-arity selector function $s_i(\vec{b}, \vec{e})$ which returns 1 iff the output of circuit $C_i$ on input $\vec{b}$ is $\vec{e}$.
It can be defined as: %
\[\forall \vec{b}, \vec{e}: \ \ s_i(\vec{b}, \vec{e}) = \prod_{k = 1}^N ( c_{i,v_k}(\vec{b}) \times id(e_k) + (1-c_{i,v_k}(\vec{b}))\times (1-id(e_k))  ). \]
Each factor of the product is 1 iff $c_{i,v_k}(\vec{b}) = e_k$. 
It has constant length, so it can be expanded using the multiplication of the logic.

We express each $q(C_i(j))$ as a function $y_i(j)$, where $\vec{b}$ is the binary representation of $j$:
\[\forall \vec{b}, \vec{e}: \ \ y_i(\vec{b}) \times s_i(\vec{b}, \vec{e}) = q(\vec{e}) \times s_i(\vec{b}, \vec{e}).\]
 
The above equation is trivially satisfied for $s_i(\vec{b}, \vec{e})=0$, thus it enforces equality of $y_i(\vec{b})$ and $q(\vec{e})$ only in the case $s_i(\vec{b}, \vec{e})=1$.
Inserting $y_i$ in the equation system~(\ref{eq:sat:cond:etr}) gives us the last $\LESOlabel$ formula:
\begin{align*} 
   & \forall \vec{b}:\ \   y_0(\vec{b}) = \sfrac{1}{8}, \ \
     y_1(\vec{b}) + y_2(\vec{b}) = y_3(\vec{b}),  \ \
     \mbox{and} \ \ y_4(\vec{b}) \times  y_5(\vec{b}) = y_6(\vec{b}),
\end{align*} 
which, combining with the previous formulas and preceded by second order existential quantifiers $\exists y_i, \exists s_i,  \exists c_{i,v}, \exists id$, with $i=0,\ldots,6$, is satisfiable if and only if the formula~\eqref{eq:sat:cond:etr} are satisfiable.
 
%
%
Obviously, the size of the resulting sentences are polynomial in the size $|C_0|+|C_1|+\myldots +|C_6|$ 
of the input instance and the sentences can be computed in polynomial time. 

This completes the construction of reduction~\eqref{prop:MC:FO:succETR:hard} and the proof of the proposition.
\end{proof}

As \LESO[0,1]{+,\times,\leq,0,1}  is weaker than $\ESOlabel_{\R}(\SUM,+,\times,\leq,<,=,\Q)$, it follows:

\begin{theorem}\label{thm:ESO:main}
Let 
$S = \R$ or $S = [a,b]$ with $[0,1] \subseteq S$,
$\{0,1\} \subseteq C \subseteq \Q$, 
$\{\times\} \subseteq O \subseteq \{+,\times,\SUM\}$ with $|O|\geq 2$,
and $E\subseteq\{\leq,<,=\}$ with $\{\leq,=\}\cap E\neq \emptyset$.
Model checking of
\begin{itemize}
\item \LESOlabel$_S(O,E,C)$ and
\item $\ESOlabel_S(O,E,C)$
\end{itemize}
is \succR-complete.
\end{theorem}
\begin{proof}
We start with the following equivalence, which follows from the fact that a comparison $a \leq b$ can be replaced by $\exists \epsilon,x: a \epsilon + x = b \epsilon$:
\begin{equation}\label{fact:<=:to:=}
\LESO[0,1]{+,\times,\leq,0,1} \equiv \LESO[0,1]{+,\times,=,0,1}.
\end{equation}

The next fact has been used by \cite{hannula2020logicsLESOandSNP}, but without proof. Perhaps the authors thought it to be too trivial to mention. But it is not obvious, since the standard technique of replacing $a \leq b$ with $\exists x: a + x^2 = b$ does not work here when $x$ is restricted to $[0,1]$.

In some sense, $\LESO[0,1]{+,\times,\leq,0,1}$ is the weakest logic one can consider in this context:

\begin{fact}\label{fact:leso:to:eso}
Let 
$S = \R$ or $S = [a,b]$ with $[0,1] \subseteq S$,
$\{0,1\} \subseteq C \subseteq \Q$, 
$\{\times\} \subseteq O \subseteq \{+,\times,\SUM\}$ with $|O|\geq 2$,
and $E \in \{\leq, =\}$.

$\LESO[0,1]{+,\times,\leq,0,1} \leq \LESOlabel_S(O,E,C) \leq \ESOlabel_S(O,E,C)$.
\end{fact}
\begin{proof}[Proof of Fact~\ref{fact:leso:to:eso}]
Relation $=$ subsumes $\leq$ due to~\eqref{fact:<=:to:=}.

If $+\in O$, the remaining statements are trivial. Otherwise, we need to express $+$ using $\SUM$.

If $S = \R$, $x + y$ can be written as $\SUM_{t} c(t)$ where $c(0) = x, c(1) = y$. (we consider model checking problems, where the finite domain can be set to binary)

If $S = [a,b]$, $x$ or $y$ might be outside the range. But the total weight of any $k$-arity function is $(b-a)^k$ and each term has a maximal polynomial degree $D$, so $x$ and $y$ are bounded by $O((b-a)^{kD})$. So all expressions can be scaled to fit in the range (Lemma 6.4. Step 3 proves this for functions that are probability distributions in \citet{hannula2020logicsLESOandSNP}).
\end{proof}
All of this combined shows the theorem. 
\end{proof}

Hannula et.~al \citep{hannula2020logicsLESOandSNP} and Durand et.~al \cite{durand2018probabilistic}
 have shown the following relationships between expressivity of the logics: $ \LESO[0,1]{+,\times,=,0,1} \le
 \LESOlabel_{d[0,1]}(\SUM,\times,=) \equiv \FOindependent{}$. 
$\LESOlabel_{d[0,1]}[O, E,C]$ means a variant of $\LESO[0,1]{O,E,C}$ where all functions are required to be distributions, 
that is 
$f^{\fA}: A^{ar(f)}\to [0,1]$ and $\sum_{\vec{a}\in A^{ar(f)}} f^{\fA}(\vec{a})=1.$
From the proof for the translation from \LESOlabel{} to  \FOindependent{}  in 
\cite{durand2018probabilistic}, it follows that the reduction 
can be done in polynomial time. Moreover, it is easy to see that model checking of 
 \FOindependent{}  can be done in \realNEXP.
 Thus we get

\begin{corollary}
Model checking of \FOindependent{} is \succR-complete.
\end{corollary}

This corollary answers the question asked in \citep{hannula2023logics} for the exact complexity of $\FOindependent$ and confirms their result that the complexity lies between $\NEXP$ and $\EXPSPACE$.

\section{Succinct \texorpdfstring{$\ETR$}{ETR} of Polynomially Many Variables}\label{sec:poly:many:variables}
The key feature that makes the language $\sumviETR$ defined in \cite{zander2023ijcai}
very powerful is the ability to index the quantified variables in the scope of summation. 
Nesting the summations 
allows handling an exponential number of variables.
Thus, similarly as in $\succETR$, sentences of $\sumviETR$ allow the use of exponentially many variables, 
however, the formulas are given directly and do not require  
any succinct encoding.
Due to the fact that variable indexing is possible,
\cite{zander2023ijcai} show that $\sumviETR$ is
polynomial time equivalent to $\succETR$.

Valiant's class $\VNP$ \citep{burgisser2000completeness,mahajan2014algebraic} is also defined in
terms of exponential sums (for completeness, we recall the definition of  $\VNP$ and related concepts in Appendix~\ref{sec:appen:alg:classes}). 
However, we cannot index variables
as above, therefore, the overall number of variables
is always bounded by the length of the defining expression.
It is natural to extend $\ETR$ with a summation operator, 
but without variable indexing as was allowed in  $\sumviETR$. 
In this way, we can have exponential sums, but the number of variables is bounded 
by the length of the formula. Instead of a summation operator, we can also
add a product operator, or both.

\begin{definition} 
\begin{enumerate}
\item $\sumETRpoly$ is defined as $\ETR$ with the addition of a unary summation operator $\sum_{x_i = 0}^1$.
\item $\prodETRpoly$ is defined similar to $\sumETRpoly$, but with the addition of a unary product operator 
    $\prod_{x_i = 0}^1$ instead. 
\item $\sumprodETRpoly$ is defined similar to $\sumETRpoly$ or $\prodETRpoly$, but including both unary summation and product operators.
\end{enumerate}
\end{definition}

 
%

In the three problems above, the number of variables is naturally
bounded by the length of the instance, since the problems are not 
succinct. 
For example, the formula $\sum_{x_1=0}^1 \sum_{x_2=0}^1 (x_1+ x_2)  (x_1+(1-x_2)) (1-x_1) = 0$ explained in the introduction is also in $\sumETRpoly$ and $\sumprodETRpoly$, but not in $\prodETRpoly$.
The formula $\sum_{e_1=0}^1 \ldots \sum_{e_N=0}^1 (x_{\langle e_1,\ldots, e_N\rangle})^2=1$ is in neither of these three classes since it uses variable indexing.

To demonstrate the expressiveness of \prodETRpoly{}, we will show that the $\PSPACE$-complete problem $\QBF$ can be reduced to it.

\begin{lemma} \label{lem:succETRpoly:1}
    $\QBF \leqp \prodETRpoly$.
\end{lemma}
\begin{proof}
    Let $Q_1x_1Q_2x_2 \dots Q_nx_n \varphi(x_1, \ldots, x_n)$ be a quantified Boolean formula with $Q_1, \ldots, Q_n \in \{\exists, \forall\}$.
    We arithmetize $\varphi$ as $A(\varphi)$ inductively using the following rules:
    \begin{description}
        \item[$\varphi$ is a variable $x_i$:] We construct $A(\varphi) = x_i$.
        \item[$\varphi$ is $\lnot \varphi_1$:] We construct $A(\varphi)$ as $1 - A(\varphi_1)$.
        \item[$\varphi$ is $\varphi_1 \land \varphi_2$:] We construct $A(\varphi)$ as $A(\varphi_1) \cdot A(\varphi_2)$.
        \item[$\varphi$ is $\varphi_1 \lor \varphi_2$:] We construct $A(\varphi)$ as $1 - (1 - A(\varphi_1)) \cdot (1 - A(\varphi_2))$ via De Morgan's law and the previous two cases.
    \end{description}
    The special treatment of the $\lor$ operator ensures that whenever $x_1, \ldots, x_n \in \{0, 1\}$, then $A(\varphi)$ evaluates to $1$ iff $x_1, \ldots, x_n$ satisfy $\varphi$ and $0$ otherwise.
    We then arithmetize the quantifiers $Q_1, \ldots, Q_n$ in a similar way, but using the unary product operator.
    \begin{description}
        \item[$Q_i = \forall$:] We construct $A(\forall x_iQ_{i+1}x_{i+1} \dots Q_n x_n \varphi)$ as $\prod_{x_i = 0}^1 A(Q_{i+1}x_{i+1} \dots Q_n x_n \varphi)$
        \item[$Q_i = \exists$:] We construct $A(\exists x_iQ_{i+1}x_{i+1} \dots Q_n x_n \varphi(x_1, \ldots, x_n)$ as\\ $1 - \prod_{x_i=0}^1 (1 - A(Q_{i+1}x_{i+1} \dots Q_n x_n \varphi))$, again using De Morgan's law.
    \end{description}

    The final $\prodETRpoly$ formula is then just $A(Q_1x_1Q_2x_2 \dots Q_nx_n \varphi) = 1$.
        
    The correctness of the construction follows because 
    a formula of the form $\forall \psi(x)$ is true over the 
    Boolean domain iff $\psi(0) \wedge \psi(1)$ is true.
    The unary product together with arithmetization allows us to write 
    the whole formula down 
    without an exponential blow-up.
\end{proof}

We 
also consider the succinct version 
of $\ETR$ with only a polynomial number of variables.

\begin{definition}
    $\succETRpoly$ is defined similar to $\succETR$, but variables are encoded in unary instead of binary, 
    thus limiting the amount of variables to a polynomial amount of variables.
    (Note that the given input circuit succinctly encodes an ETR formula and not an arbitrary circuit.)
\end{definition}

For $\succETR$, it does not matter whether the underlying structure
of the given instance is a formula or an arbitrary circuit,
since we can transform the circuit into a formula using Tseitin's trick.
This, however, requires a number of new variables 
that is proportional to the size of the circuit, which
is exponential.

Like for $\ETR$, we can now define corresponding classes
by taking the closure of the problems defined above. 
It turns out that we get meaningful classes in this way, however,
for some unexpected reason. Except for 
$\sumETRpoly$, all classes coincide with $\PSPACE$, which we will see below. 
By restricting the number of variables to be 
polynomial, the complexity of $\succETR$ reduces 
considerably, from being $\realNEXP$-complete, which contains $\NEXP$,
to $\PSPACE$. On the other hand,
the problems are most likely more powerful than $\ETR$,
assuming that $\existsR$ is a proper subset of $\PSPACE$,
which is believed by at least some researchers.

\begin{definition}
Let $\succRpoly$ be the closure of $\succETRpoly$
under polynomial time many one reductions.
\end{definition}


\begin{theorem}\label{thm:main:pspace:succETRpoly}
$\PSPACE = \succRpoly$ and the problems $\prodETRpoly$, $\sumprodETRpoly$, and $\succETRpoly$ are $\PSPACE$-complete.
\end{theorem}


\begin{proofidea}To show that $\succETRpoly$ is in $\PSPACE$,
we rely on results by \cite{renegar1992computational}.
One of the famous consequences of Renegar's work is
that $\ETR \in \PSPACE$. But Renegar shows even more,
because he can handle an exponential number of arithmetic
terms of exponential size with exponential degree as long as the number of variables is polynomially bounded. For the completeness of  
 $\prodETRpoly$, it turns out that an unbounded product is able to
 simulate an arbitrary number of Boolean quantifier alternations,
 in constrast to an unbounded sum.
So, as shown in Lemma~\ref{lem:succETRpoly:1}, we can reduce $\QBF$ to it.\end{proofidea}

\section{\texorpdfstring{$\ETR$}{ETR} with the Standard Summation Operator}\label{sec:etr:sum}

In the previous section, we have seen that $\ETR$ with 
a unary product operator ($\prodETRpoly$) is $\PSPACE$-complete.
Moreover, allowing both  unary summation and product operators 
does not lead to an increase in complexity.
In this section, we  investigate the complexity of $\sumETRpoly$,
$\ETR$ with only unary summation operators.

\begin{definition}
Let $\existsR^\Sigma$ be the closure of $\sumETRpoly$
under polynomial time many one reductions.
Moreover, for completeness, let $\existsR^\Pi$ be the closure of $\prodETRpoly$
under polynomial time many one reductions.
\end{definition}

\subsection{Machine Characterization of \texorpdfstring{$\existsR^\Sigma$}{∃ℝ\^{}Σ}}\label{sec:etr:sum:machine:char}
By Theorem~\ref{thm:main:pspace:succETRpoly},  we have $\existsR^\Pi =  \PSPACE$.
For $\existsR^\Sigma$, 
we can  conclude: 
$\realNP = \existsR \subseteq  \existsR^\Sigma \subseteq \PSPACE$
from Lemmas~\ref{lem:succETRpoly:3} and \ref{lem:succETRpoly:4}.
We conjecture that all inclusions are strict. In this section, we will provide 
some arguments in favor of this.

We first observe that using summations we can quite easily
solve $\ccPP$-problems. In particular, we have:
\begin{lemma}
    \label{lem:npPP2existRsum}
  $\NP^{\ccPP} \subseteq \existsR^\Sigma$ .
\end{lemma}
\begin{proof}
The canonical \NP$^{\ccPP}$-complete problem E-MajSat is deciding the satisfiability of a formula 
\[
  \psi: \exists x_1,\ldots, x_n:\, \# \{ (y_{1},\ldots, y_{n}) \in \{0,1\}^n \mid \phi(x,y) = 1\} \geq 2^{n-1},
\]
  i.e., deciding whether there is an assignment to the $x$-variables
  such that the resulting formula is satisfied by at least half
  of the assignments to the $y$-variables
  \citep{littman1998computational}.  

Let $ X_1\ldots X_n, Y_{1}\ldots Y_{n}$ be real variables 
and $\phi^\IR$ the arithmetization of $\phi$.
We build an equivalent $\existsR^{\Sigma}$ instance as follows:

\begin{enumerate}
\item $X_i = 0 \vee X_i = 1$, $1 \le i \le n$, and
\item $\displaystyle \sum_{Y_1=0}^1\ldots\sum_{Y_n=0}^1 
\phi^\IR(X,Y) \geq 2^{n-1}$.
\end{enumerate}

Then this instance is satisfiable iff $\psi$ is satisfiable because $X_i$ are existentially chosen and constraint to be Boolean 
and $\sum_{Y_1=0}^1\ldots\sum_{Y_n=0}^1 \phi^\IR(X,Y) $ 
is exactly the number of satisfying assignments to the $Y$-variables.
\end{proof}

Similarly to $\existsR = \realNP$, we can also characterize $\existsR^\Sigma$ using a machine model instead of a closure of a complete problem under polynomial time many one reductions.
For this we define a $\realNP^{\VNP_\IR}$ machine to be an $\realNP$ machine with 
a $\VNP_\IR$ oracle,
where $\VNP_\IR$ denotes Valiant's $\NP$ over the reals.
Since $\VNP_\IR$ is a family of polynomials, the oracle allows us to evaluate a family 
of polynomials, for example the permanent, at any real input\footnote{We recall the definitions in Appendix~\ref{sec:appen:VNP}.}.
The two lemmas below demonstrate that $\realNP^{\VNP_\IR}$ coincides with $\existsR^\Sigma$
which strengthens Lemma \ref{lem:npPP2existRsum} that $\NP^\ccPP \subseteq \existsR^\Sigma$
and characterizes $\sumETRpoly$ in terms 
of complexity classes over the reals.

\begin{lemma}\label{lemma:sumETRpoly:in:realNP:VNP}
    $\sumETRpoly \in \realNP^{\VNP_\IR}$.
    This also holds if the $\realNP$ machine is only allowed to call its oracle once.
\end{lemma}

\begin{lemma}\label{lemma:sumETRpoly:is:realNP:VNP:IR hard}
    $\sumETRpoly$ is hard for $\realNP^{\VNP_\IR}$.
\end{lemma}

\begin{theorem}\label{corollary:sumETRpoly:is:realNP:VNP:IR complete}
$\sumETRpoly$ is complete for $\realNP^{\VNP_\IR}$.
Thus, $\existsR^\Sigma = \realNP^{\VNP_\IR}$. 
\end{theorem}

\begin{proofidea}To prove Lemma~\ref{lemma:sumETRpoly:in:realNP:VNP}, we 
show a normal form for $\sumETRpoly$ instances such that 
all polynomials contained in it are of the form
$\sum_{Y \in \{0, 1\}^m} p(X, Y)$ where $p$ does not contain any unary sums.
Then we show how to translate formulas in this normal form
into a real word-RAM with oracle access.
For the hardness results of Lemma~\ref{lemma:sumETRpoly:is:realNP:VNP:IR hard},
we encode the real word-RAM computations into an ETR-instance, where the oracle calls
(which w.l.o.g.\ can be assumed to be calls to the permanent) are simulated by
the summation operator.\end{proofidea}

\subsection{Reasoning about Probabilities in Small Models}\label{sec:etr:sum:machine:sm}
In this section, we employ the satisfiability problems for languages 
of the causal hierarchy. 
The problem $\SATprobpolysumsm$
is defined like $\SATprobpolysum$, but in addition 
we require that a satisfying distribution has only polynomially
large support, that is, only polynomially many
entries in the exponentially large table of probabilities are nonzero.
Formally we can achieve this by extending an instance with an additional unary input $p \in \IN$
and requiring that the satisfying distribution has a support of size at most $p$.
The membership proofs of $\SATprobpoly$ in $\NP$ and in $\existsR$, respectively,
by \cite{fagin1990logic}, \cite{ibeling2020probabilistic}, and \cite{ibeling2022mosse} 
rely on the fact that the considered formulas have the small model property:
If the instance is satisfiable, then it is satisfiable by
a small model. For $\SATprobpolysum$, this does not seem
to be true because we can directly force any model to be arbitrarily large, e.g., 
by encoding the additional parameter $p$ above in binary or by enforcing a uniform distribution using $\sum_{x_1}\ldots\sum_{x_n} (\pp{X_1=x_1,\ldots,X_n=x_n} - \pp{X_1=0,\ldots,X_n=0})^2 = 0$. 
Thus, we have to explicitly require that the models are small, 
yielding the problem $\SATprobpolysumsm$.
Formally, we use the following:
\begin{definition}\label{lab:def:small:model:new:U}
	The decision problems $\SATprobpolysumsm$ take as input a formula 
	$\varphi \in \Lprobpolysum$  
	and a unary encoded number $p \in \IN$
	and ask whether there exists a model  $\fM=(\{X_1,\myldots,X_n\}, P)$ 
	such that $\fM \models\varphi$ and
 	$\ \#\{(x_1,\myldots,x_n): \pp{X_1=x_1,\myldots, X_n=x_n} > 0\}\le p.$
\end{definition}



It turns out that $\SATprobpolysumsm$ is a natural complete
problem for $\existsR^\Sigma$:
\begin{theorem}\label{thm:main:sec:Sigma:ETR}
 The decision problem $\SATprobpolysumsm$ is complete for $ \existsR^\Sigma$. 
\end{theorem}

\begin{proofidea}To show the containment of $\SATprobpolysumsm$ 
in $ \existsR^\Sigma$, we first show a normal form that every probability
occurring in the input instance contains all variables. Then
we have to use the exponential sum and the polynomially many variables
to ``built'' a probability distribution with polynomial support.
The lower bound follows from reducing from a restricted $\sumETRpoly$-instance.\end{proofidea}

\section{Discussion}\label{sec:landscape}

Traditionally, ETR has been used to characterize the complexity of problems
from geometry and real optimization. It has recently been used to characterize
probabilistic satisfiability problems, which play an important role
in AI,
see e.g.\ \cite{ibeling2022mosse, zander2023ijcai}. 
We have further investigated the recently defined class $\succR$,
characterized it in terms of real word-RAMs, and shown the existence of further natural
complete problems. 
Moreover, we defined a new class $\existsR^\Sigma$ and also gave natural complete problems for it.

The studied summation operators allow the encoding of exponentiation, but only with integer bases, so they do not affect the decidability, unlike Tarski's exponential function \cite{tarski1949decision}. 

Sch\"afer and Stefankovic \cite{DBLP:journals/corr/abs-2210-00571} consider
extensions of $\ETR$ where we have a constant number of alternating quantifiers
instead of just one existential quantifier. By the work of Grigoriev
and Vorobjov \cite{DBLP:journals/jsc/GrigorevVV88}, these classes are all contained
in $\PSPACE$. Can we prove a real version of Toda's theorem \citep{todasTheorem1991}?
Are these classes contained in $\existsR^\Sigma$?

\newpage
\bibliography{lit}


\newpage
\appendix

\section{Technical Details and Proofs}\label{sec:appendix}

\subsection{Proofs of Section~\ref{sec:nexp:over:the:reals}}
For a real RAM $M$, $|M|$ denotes 
the size of the encoding of $M$ (``G\"odel number'') with respect to
some standard encoding. 

\newcommand{\Powersoftwo}{\texttt{PowersOf2}}
\newcommand{\Fixinput}{\texttt{FixInput}}
\newcommand{\Wordsarewords}{\texttt{WordsAreWords}}
\newcommand{\Execute}{\texttt{Execute}}
\newcommand{\Update}{\texttt{Update}}
\newcommand{\Step}{\texttt{Step}}
\newcommand{\Perm}{\texttt{Perm}}

\begin{lemma}
    \label{lem:makesucc}
    Given a real word-RAM $M$, an input $I$,
    and $w, t \in \IN$ with $|I| \le t$, we can compute a succinct circuit $C$  
    of size $\poly(|M| \cdot |I| \cdot w) \cdot \polylog(t)$ in time $\poly(|M| \cdot |I| \cdot w) \cdot \polylog(t)$ encoding an existential formula $F$ over the reals such that $F$
    is true if $M$ accepts $I$ within $t$ steps.
\end{lemma}
\begin{proof}[Proof of Lemma~\ref{lem:makesucc}]
    The construction by \cite{erickson2022smoothing} showing that $\ETR$ is $\realNP$-complete produces a systematic formula that can easily be encoded as a succinct circuit of size polynomial in the input and word size and logarithmic in the amount of steps used. Their construction is similar to standard textbook proofs
    of Cook's Theorem.

    More specifically, the ETR formula constructed by Erickson et al.\ (in Lemma 11
    of their paper) consists
    of four main components: $\Powersoftwo$, $\Fixinput$, $\Wordsarewords$,
    and $\Execute$. Two arrays of size $2^w$ of variables are used to simulate
    the registers of $M$. The circuit $C$ gets a long bitstring as an
    input. It uses the first two bits to determine which of the four subformulas
    the node belongs.
    
    $\Powersoftwo$ constructs the constants $2^0$ to $2^w$ and makes 
    sure that they are stored in some variables. (Recall that formulas in $\ETR$
    only contain constants $0$ and $1$.) This is a small formula (of size polynomial
    in $w$) and can be output by the circuit $C$ explicitly. 

    Since formulas in $\ETR$ do not have ``word variables'', we have to ensure that
    the variables that are used to store the content of the word registers of $M$
    only contain integer values between $0$ and $2^w - 1$. This is done by 
    the formula $\Wordsarewords$. This formula is a big conjunction (of size $2^w$)
    and each term in the conjunction constrains one variable. Therefore,
    this subformula has a very regular structure and can be constructed succinctly easily.

    The formula $\Fixinput$ just makes sure that the input $I$ is hardcoded into $C$.
    This is again a small formula.

    Finally, the formula $\Execute$ is used to simulate the computation of $M$.
    This is again a big conjunction of length $t$. In the $\tau$th part,
    it is checked that the $\tau$th step of the computation is executed properly.
    This is done by a conjunction of size $|M|$. Again due to the regular structure,
    $C$ can find the desired node easily. All commands of the real RAM
    can be simulated by small formulas, except for indirect addressing.
    The simulation of indirect addressing also uses large formulas of size $2^w$,
    but again, they are a big disjunction and therefore have again a regular structure.
\end{proof}

This directly proves Lemma~\ref{lem:realNEXPsuccR}, i.e.\ that $\succETR$ is $\realNEXP$-complete, in analogy to the well-known results that
the succinct version of $3\text{-}\textsc{Sat}$ is $\NEXP$-complete.

Van der Zander et al.~\cite{zander2023ijcai} require that formulas in $\succETR$
have all negations pushed down to the bottom in the Boolean part,
i.e., they can be merged with the arithmetic terms. 
For non-succinct instances in $\ETR$, this does not make a difference,
since we can push down the negations explicitly using De Morgan's law.
For succinct instances, it is not immediately clear how to achieve
this, since a path in the Boolean part can be exponentially long.  
We here show that given a circuit $C$, an instance 
of $\succETR$ with negations in arbitrary places, 
we can transform it into a circuit $C'$ in polynomial time,
that encodes an equivalent instance of $\succETR$
with all negations pushed to the bottom.

\begin{lemma}
    \label{lem:removeneg}
    $\succETR$ is equivalent to $\succETR$ where the Boolean formula part of the circuit does not contain negations and only contains the comparison operators $<$ and $=$.
    Additionally the number of variables in the formula is not affected by this transformation.
\end{lemma}
\begin{proof}
    Let $C$ be a circuit succinctly encoding a formula.
    We show how we can locally change the structure of the formula encoded by $C$ to not contain negations or different comparison operators than $<$ or $=$.
    Each replacement gadget has a constant size, so formally we use a constant amount of bits to address which node inside the gadget we are addressing and the remaining bits are used to identify the original node.
    Since the gadgets are of differing sizes some of those nodes may be unused.

The nodes $w$ computing arithmetic value are replaced by $3$ identical copies $w_1, \ldots w_3$ of $w$ where every usage of a child $c$ of $w$ is replaced by a usage of $c_i$ in $w_i$. This step is needed to ensure that the resulting circuit encodes formulas as trees. Every node $v$ computing Boolean value is replaced by two new nodes $v^+$ and $v^-$ representing the positive and the negated value of $v$ respectively. We define the construction as follows:

    Let $v$ be a node computing an atomic Boolean value. We consider three cases:
    \begin{description}
           \item[$v$ computes $s = t$:] We compute $v^+$ as $s_1 = t_1$ and $v^-$ as $(s_2 < t_2) \lor (t_3 < s_3)$.
        \item[$v$ computes $s \leq t$:] We compute $v^+$ as $(s_1 < t_1) \lor (s_2 = t_2)$ and $v^-$ as $(t_3 < s_3)$.
        \item[$v$ computes $s < t$:] We compute $v^+$ as $s_1 < t_1$ and $v^-$ as $(t_2 < s_2) \lor (s_3 = t_3)$.
     \end{description}   
     
     Now, let $v$ be a node labeled with a Boolean operator, with children $s$ and possibly $t$.   
    We distinguish between the different operator labeling of $v$:
    \begin{description}
        \item[$v$ computes $s \lor t$:] We compute $v^+$ as $s^+ \lor t^+$ and $v^-$ as $s^- \land t^-$ via De Morgan's law.
        \item[$v$ computes $s \land t$:] We compute $v^+$ as $s^+ \land t^+$ and $v^-$ as $s^- \lor t^-$ via De Morgan's law.
        \item[$v$ computes $\lnot s$:] We compute $v^+$ as $s^- \land (0=0)$ and $v^-$ as $s^+ \land (0=0)$.
            Note that both the $0 = 0$ computations are separate fresh formulas in order to ensure that the result is a formula.
    \end{description}

    In all cases it is easy to inductively check that $v^+$ has the same Boolean value as $v$ and $v^-$ computes $\lnot v$.
    Furthermore every node only has one parent, making the construction a formula.
    All these constructions are local and thus can easily be succinctly encoded by using the original circuit $C$.
\end{proof}

\subsection{Proofs of Section \ref{sec:relationships:bool:class:vs:ove:the:reals}}

\begin{proof}[Proof of Theorem~\ref{thm:translation:lower}]
    We always have $\succR \supseteq \NEXP$.

    Now we show $\succETR \in \NEXP$ assuming $\existsR = \NP$ which implies $\succR = \NEXP$.
    $\existsR = \NP$ implies that there is a polytime NTM $N$ deciding $\ETR$.
    We construct an NTM $N'$ as follows:
    \begin{description}
        \item
            Given a circuit $C$ representing a succinct representation of an existential formula $\varphi$ over the reals, reconstruct $\varphi$ by evaluating $C$ at all possible inputs.
            Then check if $\varphi \in \ETR$ by using $N$.
    \end{description}
    Note that $\varphi$ has size at most exponential in $|C|$ and as such the usage of $N$ incurs at most a running time exponential in $|C|$.
    Thus $N'$ nondeterministically decides $\succETR$ in time exponential in $|C|$ and we have $\succETR \in \NEXP$.
\end{proof}

\begin{proof}[Proof of Theorem~\ref{thm:translation:upper}]
    We always have $\succR \subseteq \EXPSPACE$.

    Let $A$ be any language in $\EXPSPACE$ decidable in some time $t(n)$.
    Then the padded language $A' = \{x\$1^{t(|x|)} \mid x \in A\}$ is in $\PSPACE$, where $\$$ is a symbol not used in $A$.
    Under the assumption $\existsR = \PSPACE$ there thus is a polytime reduction $f$ from $A'$ to $\ETR$.

    We construct a nondeterministic real RAM $M$ as follows:
    \begin{description}
        \item On input $x$, construct the formula $\varphi = f(x\$1^{t(|x|)})$.
            Then nondeterministically guess an assignment to all the existentially quantified variables and verify if $\varphi$ is satisfied.
    \end{description}
We can guess a Boolean assignment by adding the constraints $x(x-1) = 0$ for each variable. Thus we have shown that $A \in \realNEXP$.
Now the claim follows by Lemma~\ref{lem:realNEXPsuccR} and the
fact that $A$ was arbitrary.
\end{proof}

Let $M_1, M_2, \ldots$ be a fixed efficiently computable enumeration of all nondeterministic real word RAMs.
Efficient here means that we can in time $O(i)$ construct a list of all instructions of $M_i$ and store each instruction in $O(1)$ cells of word size $O(\log i)$.
We can then access the $j$-th instruction in constant time.
\begin{lemma}
    \label{lem:nondetsim}
    There is a nondeterministic real word-RAM $U_N$ that given $i, w, t \in \IN$ and a binary input $x \in \{0, 1\}^\star$ can simulate running $M_i$ on $x$ with word-size $w$ for $t$ steps in time $T_N \in O(i + t + |x|)$ using any word-size $w_N \geq w + O(\log i + \log t)$.
\end{lemma}
\begin{proof}[Proof of Lemma \ref{lem:nondetsim}]
    We construct $U_N$ by splitting the memory into two sections, a section $A$ containing $2^w$ uninitialized real cells and word cells each to store the memory of $M_i$ and a section $B$ of size $O(i + \log t)$ to store the instructions of $M_i$, the index of the next instruction to be simulated and a counter storing the number of steps left to be simulated, as well as any temporary variables needed to simulate the operations.
    At the start $U_N$ copies $x$ into the first $|x|$ word cells of section $A$ of the memory and initializes section $B$ to contain the instructions of $M_i$, the index of the first instruction and the number of steps left as $t$.
    This initialization takes time $O(|x| + i + \log t)$.
    Then $U_N$ simulates $M_i$ step by step as long as there are any steps left.
    It takes time $O(1)$ to find the next instruction to be simulated.
    Any access to memory that $M_i$ uses is translated to address into section $A$ instead and every operation is modified to behave as if the cell had word-size $w$.
    Afterwards the counter for the steps remaining is decremented.
    All of this is possible in (amortized) time $O(1)$, so in total the simulation itself takes $O(t)$ steps.

    The only time any of the cells needs a word-size (potentially) larger than $w$ is to address both sections $A$ and $B$ of the memory, to store the instructions of $M_i$ or to store the step-counter.
    For all of these a word-size of $w + O(\log i + \log t)$ is sufficient.

    There is no need for any special treatment of nondeterminism, since in this model uninitialized cells automatically are treated as nondeterministically initialized cells.
\end{proof}

We also need to simulate nondeterministic real RAMs using a 
deterministic real RAM.
In the Boolean case, we could simply iterate over all proof strings.
This does not work in the real case, since we cannot iterate over
all real proof strings. We use an algorithm of Renegar
instead (see Lemma~\ref{lem:renegar}), which is a Boolean
algorithm to decide whether an instance of $\ETR$ is true.
The algorithm is stated as a parallel algorithm, since Renegar used
it to show that $\ETR \in \PSPACE$ (using the fact that parallel polynomial time is polynomial space). It can be transformed into a 
sequential algorithm using standard methods.

\begin{lemma}
    \label{lem:detsim}
    There is a deterministic real word-RAM $U_D$ that given $i, w, t \in \IN$ and a binary input $x \in \{0, 1\}^\star$ can simulate running $M_i$ on $x$ with word-size $w$ for $t$ steps in time $T_D \in 2^{\poly(|x|, i, 2^w, t)}$ using any word-size $w_D \geq \poly(\log |x|, \log i, w, \log t)$.
\end{lemma}
\begin{proof}[Proof of Lemma \ref{lem:detsim}]
    We use \cite{erickson2022smoothing} to construct an existential formula $F$ over the reals for $M_i$ with input $x$ being simulated with word-size $w$ for $t$ steps of size $\poly(|x|, i, 2^w, t)$.
    We can then use Renegar's algorithm (see Lemma~\ref{lem:renegar}) to deterministically check whether $F$ is satisfiable in time $2^{\poly(|x|, i, w, t)}$ for any word-size $w_D \geq \poly(\log |x|, \log i, w, \log t)$.
    Note that Renegar's algorithm does not need any of the real cells whatsoever and only needs the word-size to address its memory.
\end{proof}

We call a function $t: \IN \to \IN$ time constructible
if there is a deterministic real RAM that on every input of length $n$ outputs $t(n)$ in time $O(t(n))$
with word size $O(\log t(n))$.

Note that the following nondeterministic time hierarchy for the real word-RAMs requires the use of $O$-Notation for both classes due to a lack of acceleration theorems for word-RAMs.
\begin{lemma}\label{lemma:proper:incl}
    Let $t_1, t_2: \IN \to \IN$ be monotone and 
    time-constructible functions with $t_1(n+1) \in o(t_2(n))$ and $t_2(n) \in \Omega(n)$.
    Then $\NTimereal(O(t_1(n))) \subsetneq \NTimereal(O(t_2(n)))$.
\end{lemma}
\begin{proof}[Proof of Lemma \ref{lemma:proper:incl}]
    Note that $\NTimereal(O(t_1(n))) \subseteq \NTimereal(O(t_2(n)))$ holds directly due to the monotonicity of $t_1$.\footnote{This condition is implicit in the original nondeterministic time-hierarchy proof in \cite{cook1972hierarchy}, by virtue of that proof being only for polynomial time bounds. \cite{seiferas1978separating} and \cite{vzak1983turing} resolve this by phrasing their results as existence of a language in the set intersection and not as a strict inclusion. \cite{vzak1983turing} requires $t_2(n+1) \in \Theta(t_2(n))$ whenever they phrase their result as a strict inclusion. Otherwise a simple counterexample where $\NTimereal(O(t_1(n))) \subseteq \NTimereal(O(t_2(n)))$ is violated would be 
    \begin{align*}
        t_1(n) &= \begin{cases}n & \text{if $n$ is even}\\n^3 & \text{otherwise}\end{cases} & t_2(n) &= \begin{cases}n^2 & \text{if $n$ is odd}\\n^4 & \text{otherwise}\end{cases}
    \end{align*}
    }
    To show non-equality we use a modification of the diagonalization in \cite{vzak1983turing}.
    Special care has to be taken in dealing with the word-sizes.

    Construct a nondeterministic real word-RAM $M$ as follows:
    \begin{itemize}
        \item[] On input $1^i01^m01^k$ of length $n = 2 + i + m + k$ compute upper bounds $T := T_D$ and $W := W_D$ for Lemma~\ref{lem:detsim} with $|x| = 2+i+m, i=i, w = \log(t_2(2+i+m))$ and $t=t_2(2+i+m)$.
        Should these values be too big to compute in time $t_2(n)$ and with word-size $\log(t_2(n))$, abort the computation and assume we are always in the first case of what follows.
            \begin{description}
                \item[if $k < T$ or $\log(t_2(n)) < W$,] then simulate $M_i$ nondeterministically on $1^i01^m01^{k+1}$ for $t_2(n)$ steps and word-size $\log (t_2(n))$ using $U_N$ from Lemma~\ref{lem:nondetsim}.
                    Then return the same result as that simulation.
                \item[if $k \geq T$ and $\log(t_2(n)) \geq W$,] then simulate $M_i$ deterministically on $1^i01^m0$ for $t_2(2+i+m)$ steps and word-size $\log(t_2(2+i+m))$ using $U_D$ from Lemma~\ref{lem:detsim}.
                    Then invert the result of that simulation.
            \end{description}
    \end{itemize}
    We see that $M$ uses time $O(t_2(n))$ in both cases.
    In the first case this follows directly from the time-constructability of $t_2$ and for the second case it follows from the fact that we only run this case if we can guarantee it will finish in time $O(t_2(n))$.
    Furthermore $M$ uses a word-size of $O(\log (t_2(n)))$ and will produce the same outputs for any higher word-size.
    For the first case this follows directly from Lemma~\ref{lem:nondetsim} and $(\log(t_2(n)) + O(\log i + \log t_2(n)) = O(\log t_2(n))$ due to $i \in O(n)$ and $t_2(n) \in \Omega(n)$.
    For the second case it follows from the fact that we only run this case if we can guarantee that $U_D$ uses a word-size of at most $W \leq \log {t_2(n)}$ and at most $T \leq k < n \in O(t_2(n))$ steps.
    Thus the language $L$ decided by $M$ is contained in $\NTimereal(O(t_2(n)))$.

    If we assume $\NTimereal(O(t_1(n))) = \NTimereal(O(t_2(n)))$, there is also an equivalent nondeterministic real word-RAM $M'$ running in time $c_1 \cdot t_1(n) + c_1$ for some $c_1 \in \IN$, being correct for every word-size $w \geq c_2 \log(c_1 \cdot t_1(n) + c_1) + c_2$ for some $c_2 \in \IN$.
    Let $i \in \IN$ be such that $M_i = M'$ and let $m \in \IN$ be such that 
    \begin{align}
        c_1 \cdot t_1(2+i+m+k+1) + c_1 &\leq t_2(2+i+m+k)\label{eq:hier:enoughtime}\\
        \text{and}\quad c_2 \log(c_1 \cdot t_1(2+i+m+k+1) + c_1) + c_2 &\leq \log(t_2(2+i+m+k))\label{eq:hier:enoughbits}
    \end{align}
    hold for all $k \in \IN$.
    Such an $m$ exists since $t_1(n+1) \in o(t_2(n))$.
    \eqref{eq:hier:enoughtime} ensures that the nondeterministic simulation of $M_i$ is run to completion and $\eqref{eq:hier:enoughbits}$ ensures its result is correct by running the simulation with a big enough word-size.
    Let $T$ and $W$ be defined as in $M$ and let $K \in \IN$ be the smallest integer with $K \geq T$ and $\log(t_2(2+i+m+K)) \geq W$ and such that the computations of $T$ and $W$ succeed in time $t_2(n)$ and with word-size $\log(t_2(n))$.
    Then we get a contradiction via
    \begin{align*}
         & \text{$1^i01^m0$ is accepted by $M_i$ with word-size $\log(t_2(2+i+m))$}\\
        \Leftrightarrow & \text{$1^i01^m01$ is accepted by $M_i$ with word-size $\log(t_2(2+i+m+1))$}\\
        \Leftrightarrow & \text{$1^i01^m011$ is accepted by $M_i$ with word-size $\log(t_2(2+i+m+2))$}\\
        \vdots\\
        \Leftrightarrow & \text{$1^i01^m01^K$ is accepted by $M_i$ with word-size $\log(t_2(2+i+m+K)) \geq W$}\\
        \Leftrightarrow & \text{$1^i01^m0$ is not accepted by $M_i$ with word-size $\log(t_2(2+i+m))$}
    \end{align*}
Thus the inclusion has to be strict.    
\end{proof}

Corollary~\ref{cor:existsR_hierarchy} then directly follows.

\subsection{Proofs of Section \ref{sec:poly:many:variables}}

It turns out that the product operator is responsible for
bringing the complexity of $\ETR$ up to $\PSPACE$.
It allows to reduce the $\PSPACE$-complete problem $\QBF$ to
$\prodETRpoly$ (see Lemma~\ref{lem:succETRpoly:1}).

In the rest of this section, we prove that $\prodETRpoly$ and the remaining problems discussed above are in $\PSPACE$.
We first reduce the problems to $\succETRpoly$
starting with obvious relations:
\begin{lemma} \label{lem:succETRpoly:2}
    $\sumETRpoly \leqp \sumprodETRpoly$ and $\prodETRpoly \leqp \sumprodETRpoly$.
\end{lemma}

\begin{lemma} \label{lem:succETRpoly:3}
    $\sumprodETRpoly \leqp \succETRpoly$.
\end{lemma}
\begin{proof}
    Let $\varphi$ be a $\sumprodETRpoly$ instance, encoded as a tree $T$.
    We construct a circuit $C$ encoding $\varphi$ by expanding the summation and product operators explicitly, i.e., for a node $v \in T$ on depth $\ell$, we construct $2^\ell$ copies of $v$, indexed as $v_{(e_1, \ldots, e_\ell)}$ with $e_i \in \{0, 1\}$. Only summation and product operators are counted for the depth $\ell$, thus the root node and all nodes in an instance without summation or product operators would have depth $0$, and an index $v_{()}$.
    
    If $v$ is a leaf, then  $v_{(e_1, \ldots, e_\ell)}$ retains the labeling of $v$, unless the labeling is one of the summation variables of its ancestors.
    In that case it is replaced by the $e_i$ corresponding to 
    the summation variable.
    Otherwise if $v$ is an inner node with children $s$ and possibly $t$ has any labeling other than the unary sum or product, then $v_{(e_1, \ldots, e_\ell)}$ has children $s_{(e_1, \ldots, e_\ell)}$ and $t_{(e_1, \ldots, e_\ell)}$.

    It remains to construct the case when $v$ is a unary sum or product.
    If $v$ is a sum gate, then we construct $v_{(e_1, \ldots, e_\ell)}$ as the sum of $v_{(e_1, \ldots, e_\ell, 0)}$ and $v_{(e_1, \ldots, e_\ell, 1)}$. If $v$ is a product gate, then we construct $v_{(e_1, \ldots, e_\ell)}$ as the product of $v_{(e_1, \ldots, e_\ell, 0)}$ and $v_{(e_1, \ldots, e_\ell, 1)}$.

    The node $v_{()}$ where $v$ is the root of $T$ is the root of the constructed $\succETRpoly$ instance.
    All of this can easily be computed by a circuit $C$ taking as input an encoding of $v$ and an encoding of $(e_1, \ldots, e_\ell)$ and an encoding of some $k \in \IN$ denoting which node of the binary tree we are in.
    All of these can be encoded with polynomially many bits in the size of $T$ and the bounds used on the unary sum and product gates.
\end{proof}

To show that $\succETRpoly$ is in $\PSPACE$,
we use results by \cite{renegar1992computational}.
One of the famous consequences of Renegar's work is
that $\ETR \in \PSPACE$. But Renegar shows even more,
because he can handle an exponential number of arithmetic
terms of exponential size with exponential degree as long as the number of variables is polynomially 
bounded.

Renegar's states his results in terms of parallel algorithms.
However, it is well-known that parallel polynomial time equals
polynomial space.

\begin{lemma}[Theorem 1.1 in \cite{renegar1992computational}]
    \label{lem:renegar}
    Let $p_1, \ldots, p_m$ be polynomials in $n$ variables $x_1, \ldots, x_n$ and of degree at most $d \geq 2$ where all coefficients are given as integers of bit length $L$.
    Additionally let $\circ_1, \ldots, \circ_m \in \{>, <, =\}$ be comparison operators and let $P$ be a Boolean function with $m$ inputs.
    Then there is an algorithm deciding whether there are values $x_1, \ldots, x_n \in \R$ such that $P(p_1(x_1, \ldots, x_n) \circ_1 0, \ldots, p_m(x_1, \ldots, x_n) \circ_m 0) = 1$ in time
    \[
        L\log(L)\log\log(L)\left(md\right)^{O(n)}
    \]
    using $\left(md\right)^{O(n)}$ evaluations of $P$.
    
    Furthermore there is a parallel algorithm deciding the same in time
    \[
        \log(L)\left(n\log(md)\right)^{O(1)} + Time(P, k)
    \]
    using $L^2(md)^{O(n)}$ processors and requiring $(md)^{O(n)}$ evaluations of $P$.
    Here $Time(P, k)$ is the parallel time needed for the evaluations of $P$ when using $k (md)^{O(n)}$ processors (for any $k \geq 1$).
\end{lemma}

\begin{lemma} \label{lem:succETRpoly:4}
     $\succETRpoly \in \PSPACE$.
\end{lemma}
\begin{proof}[Proof of Lemma \ref{lem:succETRpoly:4}]
    We construct a parallel polynomial time algorithm for $\succETRpoly$:
    Let $C$ be a circuit computing a function $f: \{0,1\}^N \to \{0,1\}^M$ succinctly encoding an ETR-formula $\varphi$ containing the variables $x_1, \ldots, x_n$.
    Construct $\varphi$ explicitly in parallel using $2^N$ processors by evaluating $f$ at each possible input.
    During this construction also check that $f$ encodes a valid ETR-formula, i.e., that all nodes agree with their parents and children, are connected to the root (unless marked as disabled) and there is no Boolean node as a child of an arithmetic node. 
    All checks can be done locally in parallel.
    The only check here that requires some care is the connectivity to the root, but we can use standard skip list techniques to do this in parallel.

    W.l.o.g.\ we may assume that all comparisons in $\varphi$ are comparisons with $0$ and only use the comparison operators $<$ and $=$. If this is not the case, then we can obtain a formula of this form by using Lemma~\ref{lem:removeneg} and replacing $a \circ b$ by $a - b \circ 0$ for $\circ \in \{<, =\}$.

    Consider the forest obtained by ignoring all Boolean nodes.
    Each root in this forest corresponds to one of the polynomials in $\varphi$.
    Use the tree contraction algorithm described in \cite{DBLP:books/el/leeuwen90/KarpR90} to compute each of the polynomials in parallel.
    We apply the tree contraction over the semi-ring of multivariate polynomials, given as lists of coefficients, and compute additions and multiplications of elements using the respective parallel algorithms.
    We finish the algorithm off by using the parallel algorithm from Lemma~\ref{lem:renegar}.
    Whenever this algorithm needs to evaluate the Boolean part of the formula, we again use the standard parallel expression evaluation algorithm using $k = 2^{O(N)}$ processors each.

    Note that due to the exponential size of $\varphi$ we have $m, d, L \leq 2^N$.
    Additionally in $\succETRpoly$ we have $n \in \poly(N)$ and thus the overall algorithm has a parallel running time of $\poly(N)$ using $2^{\poly(N)}$ processors, implying $\succETRpoly \in \PSPACE$.
\end{proof}

Now Lemmas \ref{lem:succETRpoly:1}, \ref{lem:succETRpoly:2},
\ref{lem:succETRpoly:3}, and \ref{lem:succETRpoly:4} together prove Theorem~\ref{thm:main:pspace:succETRpoly}.

\subsection{Proofs of Section \ref{sec:etr:sum:machine:char}}

Similarly to $\existsR = \realNP$, we can also characterize $\existsR^\Sigma$ using a machine model instead of a closure of a complete problem under polynomial time many one reductions.
For this we define a $\realNP^{\VNP_\IR}$ machine to be an $\realNP$ machine with 
a $\VNP_\IR$ oracle,
where $\VNP_\IR$ denotes the families of polynomials in $\VNP$ over the reals.
Since $\VNP_\IR$ is a set of polynomials, the oracle allows us to evaluate a family 
of polynomials, for example the permanent, at any real input\footnote{We recall the used definitions in Appendix~\ref{sec:appen:VNP}.}.
In order to do so, we add the following operation to a real word-RAM:
\[
    R[i] \gets O(R[W[j]]\dots R[W[k]]).
\]
Recall that for a real word-RAM $R[i]$ denotes a real register, while $W[i]$ denotes a word register.
The values $i, j, k$ are assumed to be constant indices.
This call evaluates the oracle $O$  by choosing the polynomial in the family 
with $W[k] - W[j] + 1$ variables and  evaluating it with the values of $R[W[j]]$ 
up to $R[W[k]]$. The result is stored in $R[i]$.
In case $W[k] > W[j]$ or the oracle family doesn't contain a polynomial with $W[k] - W[j] + 1$ variables\footnote{This can for example happen for the permanent, where the number of variables is always a perfect square.}, the machine rejects.

\begin{proof}[Proof of Lemma~\ref{lemma:sumETRpoly:in:realNP:VNP}]
    This proof proceeds in two stages:
    \begin{enumerate}
        \item \label{first_step} We normalize the $\sumETRpoly$ formula such that all polynomials contained in it have the form $\sum_{Y \in \{0, 1\}^m} p(X, Y)$ where $p$ does not contain any unary sums, then
        \item \label{last_step} translate the formula into a real word-RAM with oracle access.
    \end{enumerate}

    To normalize the formula we introduce variables $Z_i = \frac{1}{2^i}$ and iteratively move the unary sums outwards by using the following replacements:
    \begin{align*}
        \sum_{Y \in \{0,1\}^m} p(X, Y) \cdot \sum_{Y' \in \{0,1\}^{m'}} p'(X, Y') &\mapsto \sum_{Y \in \{0,1\}^m} \sum_{Y' \in \{0,1\}^{m'}} p(X, Y) \cdot p'(X, Y')\\
        \sum_{Y \in \{0,1\}^m} p(X, Y) + \sum_{Y' \in \{0,1\}^{m'}} p'(X, Y') &\mapsto \sum_{Y \in \{0,1\}^m} \sum_{Y' \in \{0,1\}^{m'}} (p(X, Y) \cdot Z_{m'} +  p'(X, Y') \cdot Z_{m})
    \end{align*}
    Note that if we included the correcting terms $\frac{1}{2^i}$ directly instead of via $Z_i$ these replacements would not be of polynomial size after moving all sums.
    A real word-RAM can now simply evaluate the formula and whenever it needs to evaluate some sum $\sum_{Y \in \{0, 1\}^m} p(X, Y)$ it queries the oracle\footnote{technically all queries have to be to polynomials from the same family, but we can query for any of these using queries to the permanent instead}.
    Since $p$ doesn't contain any unary sums anymore, this polynomial is in $\VNP_\IR$.

    To reduce the number of oracle queries we add another two steps before step 1, first turning the formula into a collection of polynomials that has a common zero iff the formula is satisfiable and then turning this collection of polynomials back into a $\sumETRpoly$-formula with one polynomial, thus only requiring a single oracle call.
    This proof follows alongside a similar proof by \cite{schaefer2017fixed} for $\ETR$ using Tseitin's trick, with the difference that we cannot use Tseitin's trick within any of the unary sums and thus the final formula does not have degree $4$.
    We require all negations in the formula to be pushed towards and subsequently absorbed into the comparison nodes and that all comparisons are represented using only $=$ and $<$, but this can be easily achieved.
    We add variables $\zeta_v$ for each node $v$ in the formula with a Boolean result.
    If they have value $0$, then the sub-formula rooted at $v$ is satisfied.
    The converse doesn't have to hold, but since the Boolean part of the formula is monotone this is fine.
    \begin{itemize}
        \item If $v$ is computing $p_1 = p_2$, then we add the polynomial $\xi_v - (p_1 - p_2)$.
        \item If $v$ is computing $p_1 < p_2$, then we add the polynomial $(p_2 - p_1)\xi_v'^2 - (1 - \xi_v)$ and a new variable $\xi_v'$.
        \item If $v$ is a Boolean or of the nodes $u$ and $w$, then we add the polynomial $\xi_v - \xi_u \cdot \zeta_w$.
        \item If $v$ is a Boolean and of the nodes $u$ and $w$, then we add the polynomial $\xi_v - (1 - (1 - \xi_u) \cdot (1 - \xi_w))$.
    \end{itemize}
    If $p_1, \ldots, p_r$ is a set of polynomials, then $\sum_{i=1}^r p_r^2 = 0$ is satisfiable iff $p_1, \ldots, p_r$ have a common zero.
    Continuing with steps \ref{first_step} until \ref{last_step} from here on out will then result in a single oracle query containing the entire formula.
\end{proof}

\begin{proof}[Proof of Lemma~\ref{lemma:sumETRpoly:is:realNP:VNP:IR hard}]
    Since the permanent is $\VNP_\IR$ complete we w.l.o.g.\ restrict ourselves to $\realNP^\per$.
    We use the same construction that \cite{erickson2022smoothing} used to show that $\ETR$ is $\realNP$ hard and extend the $\Execute$ component to allow for the oracle queries.
    In order to do so we quickly repeat the relevant parts of their definitions:
    \begin{itemize}
        \item For each address $i$, variable $\llbracket W(i,t) \rrbracket$ stores the value of word register $W[i]$ at time $t$.
        \item For each address $i$, variable $\llbracket R(i,t) \rrbracket$ stores the value of real register $R[i]$ at time $t$.
        \item Variable $\llbracket pc(t) \rrbracket$ stores the value of the program counter at time $t$.
    \end{itemize}

    The formula $\Execute$ is defined via an $\Update(t, \ell)$ constraint that, given the content of the memory at time $t-1$, enforces the correct values for the memory at time $t$ when executing instruction $\ell$ in the program code.
    If $L$ is the total number of instruction in the program and $T$ the total number of steps to simulate the program for, then
    \[
        \Execute := \bigwedge_{t=1}^T\bigwedge_{\ell=1}^L\left((\llbracket pc(t) \rrbracket = l) \implies \Update(t, \ell)\right)\,.
    \]
    All $\Update(t, \ell)$ that do not influence the program flow always contain a $\Step(t)$ constraint that simply advances the program counter:
    \[
        \Step(t) := (\llbracket pc(t) \rrbracket = \llbracket pc(t-1) \rrbracket +1)
    \]
    We implement $\Update(t, \ell)$ for oracle query instructions $R[i] \gets O(R[W[j]]..R[W[k]])$ as follows:
    \begin{align*}
        \Step(t) \land \bigvee_{a,b} \left((\llbracket W(j, t-1) \rrbracket = a) \land (\llbracket W(j, t-1) \rrbracket = b) \land \llbracket R(i, t) \rrbracket = \Perm(t, a, b)\right)
    \end{align*}
    where the big or is going over all $a, b \in \{0, \ldots, 2^w-1\}$ such that $a \leq b$ and $b-a+1$ is a square number, and $\Perm(t, a, b)$ denotes an expression that evaluates the permanent on $\llbracket R(a, t-1) \rrbracket, \ldots, \llbracket R(b, t-1) \rrbracket$.
    Remains to show how to evaluate the permanent.
    Let $m = \sqrt{b-a+1}$, then we can evaluate the permanent using
    \begin{align*}
        \Perm(t, a, b) := \sum_{M \in \{0,1\}^{m \times m}} \delta(M) \cdot \prod_{\mu=1}^m\sum_{\nu=1}^m M_{\mu\nu} \cdot \llbracket R(a+\mu\cdot m + \nu, t-1) \rrbracket
    \end{align*}
    where $\delta(M)$ is a polynomial that is $1$ iff $M$ is a permutation matrix and $0$ otherwise.
    Given a polynomial $p(x)$ that is $1$ if $x = 1$ and $0$ for $x \in \{0, \ldots, m\} \setminus \{1\}$, we can use the characterization that a Boolean matrix is a permutation matrix iff each row and each column contains exactly one $1$:
    \[
        \delta(M) = \prod_{\mu=1}^m p\left(\sum_{\nu=1}^m M_{\mu\nu}\right) \cdot \prod_{\nu=1}^m p\left(\sum_{\mu=1}^m M_{\mu\nu}\right)
    \]
\end{proof}

\subsection{Proofs of Section \ref{sec:etr:sum:machine:sm}}
We show that $\SATprobpolysumsm$ is contained in $\existsR^\Sigma$,
the first half of the completeness proof of Theorem~\ref{thm:main:sec:Sigma:ETR}.

\begin{lemma}\label{lemm:prob:sum:etr:sum}
$\SATprobpolysumsm \in \existsR^\Sigma$. 
\end{lemma}

\begin{proof} 
    Given a $\SATprobpolysumsm$ instance $\varphi$, we first want to ensure
    that all primitives are of the form $\PP{X_1 = \xi_1,\dots,X_N = \xi_n}$,
    that is, every variable occurs in the primitive
    and everything is connected by conjunctions.
    This is achieved inductively. 
    So let $\PP{\delta}$ be given. We define the arithmetization
    $A_\delta(\xi_1,\dots,\xi_n)$ of $\delta$ to be $1$,
    if $\PP{X_1 = \xi_1,\dots,X_N = \xi_n}$
    contributes to $\PP{\delta}$, and $0$ otherwise. 
    $A_\delta$ can be easily derived from $\delta$ by following the
    logical structure of $\delta$. We first present the construction
    for the Boolean domain for simplicity.
    \begin{itemize}
        \item In the base case $\delta$ equals $X_i = \xi$. The corresponding arithmetization
        is $1-(X_i - \xi)^2$. 
        \item If $\delta = \delta_1 \wedge \delta_2$, then 
        $A_\delta = A_{\delta_1} A_{\delta_2}$.
        \item If $\delta = \neg \delta_1$, then $A_\delta = 1 - A_{\delta_1}$
    \end{itemize}
    Since we implement $A_\delta$
    as a polynomial, it is also easy to implement the starting 
    predicates ``$X_i = \xi$'' over larger domains $D$.
    It is simply a polynomial of degree $|D|$ that is $1$ on $\xi$
    and $0$ everywhere else.
    Then we have 
    \[
      \PP{\delta} = \sum_{\xi_1} \dots \sum_{\xi_n} A_\delta(\xi_1,\dots,\xi_n)
        \PP{X_1 = \xi_1,\dots,X_n = \xi_n}
    \]
    Note that the righthand side is expressible as a 
    $\SATprobpolysumsm$ instance.
    
    For the reduction, we now want to replace the primitives 
    $\PP{X_1 = \xi_1,\dots,X_N = \xi_n}$ by variables. 
    There are exponentially many primitives,
    but we only have polynomially many variables.
    However, we are only interested in small models,
    so we only need to simulate exponentially many variables
    with small support, say of size $p$.
    The variables $X_1,\dots,X_p$, will store the values.
    Selector variables $e_{i,j}$, $1 \le i \le p$, $1 \le j \le n$,
    will assign them to places in our big array.
    We quantify as follows:
    \begin{align*}
      \exists X_1, \dots, X_p \exists e_{1,1},\dots,e_{p,n} 
      \bigwedge_{i,j} \prod_{d \in D} (e_{i,j} - d) = 0   \dots
    \end{align*}
    This constrains the $e_{i,j}$ to have values in $D$.
    Then the expression 
    \[
       X(f_1,\dots,f_n) = \sum_{i = 1}^p X_i \prod_{j = 1}^n (1 - (f_j - e_{i,j})^2)
    \]
    with $f_{i} \in D$ simulate exponentially many variables with support size $p$.  
    We can replace the $\PP{X_1 = \xi_1,\dots,X_n = \xi_n}$
    by $X(\xi_1,\dots,\xi_n)$ and get that 
    $\SATprobpolysumsm$ is in $\existsR^\Sigma$.
    
    For larger domains $D$, we just have to replace the expressions
    $1-(X_i - \xi)^2$ or $(f_j - e_{i,j})^2$ by a polynomial, that is exactly $1$
    on the element $\xi$ or $e_{i,j}$ and $0$ on all other elements of $D$.
    Such a polynomial can be easily found using Lagrange interpolation.
\end{proof}

The second half of the completeness result for $\SATprobpolysumsm$
is given in the following lemma.

\begin{lemma}\label{lemma:sumETRpoly:leqp:SATprobpolysumsm}
$\sumETRpoly \leqp \SATprobpolysumsm$.
\end{lemma}
Like before, let $\sumETRpoly_{1/2}$ denote the following problem:
Instances are the same as for $\sumETRpoly$.
However, we ask the question whether there is a solution
with the sum of the absolute values ($\ell_1$ norm)
being bounded by $1/2$. We prove that the general 
$\sumETRpoly$ can be reduced to this problem. 
\cite{abrahamsen2018art} prove similar results for
standard $\ETR$.

Using this intermediate problem we prove $\sumETRpoly \leqp \SATprobpolysumsm$, i.e.\ Lemma~\ref{lemma:sumETRpoly:leqp:SATprobpolysumsm}, in two steps.

\begin{lemma}
\label{lem:sumETR_eight_hardness}
$\sumETRpoly \leqp \sumETRpoly_{1/2}$.
\end{lemma}

The proof of the first lemma is almost identical to 
the proof of Lemma~\ref{lem:sumvietr1:reduction} and therefore omitted.

\begin{lemma}\label{lemma:sumETRpoly:1:8:leqp:SATprobpolysumsm}
$\sumETRpoly_{1/2} \leqp \SATprobpolysumsm$.
\end{lemma}
\begin{proof}
    Assume that our $\sumETRpoly_{1/2}$ instance $\varphi$ uses
    variables $x_1,\dots,x_n$ and summation variables
    $e_1,\dots,e_m$.
    
    In the $\SATprobpolysumsm$ instance, we will have ternary random variables $X_1,\dots,X_n$ and $E$.
    $x_i$ will be represented by the term $p_i := \PP{X_1 = 0, \dots, X_i = 1, \dots X_n = 0, E = 0} - \PP{X_1 = 0, \dots, X_i = -1, \dots X_n = 0, E = 0}$.
    $e_i$ will be represented by the term $q_{e} :=\PP{X_1 = 0, \dots,X_n = 0, E = e}$ for $e \in \{0, 1\}$.
    The constants $m$ and $n$ can be easily constructed from $\PP{\top}$, which is $1$.
    
    First, we enforce $q_{0} = 0$ and $q_{1} = 1/2$.
    (We use fractions for convenience. We can always multiply by the common denominator.)
    This places total mass $1/2$ on one entry in the distribution.
    Now a summation $\sum_{e_i = 0}^1$ in $\varphi$ can be simulated by the summation operator $\sum_{e = 0}^1$ of $\SATprobpolysumsm$.
    Inside the summation, the term $2 \cdot q_{e}$ simulates the term $e_i$ in the summation of $\varphi$.
    
    The variables $x_i$ will be simulated by the expressions $p_i$.
    By choosing $$\PP{X_1 = 0, \dots, X_i = 1, \dots X_n = 0, E = 0} = \min(0, x_i)$$ and $$\PP{X_1 = 0, \dots, X_i = -1, \dots X_n = 0, E = 0} = -\max(0, x_i),$$ combined with the fact that $\varphi \in \sumETRpoly_{1/2}$ always has a solution with $\ell_1$ norm bounded by $1/2$, at most $1/2$ of the probability mass are placed on the $p_i$ in total.
    Vice-versa to express any solution with $\ell_1$ norm bigger than $1/2$ would require more than $1/2$ probability mass being put on the $p_i$, in turn making it not a valid probability distribution when added to the $1/2$ probability mass on $q_1$.
    
    Now we can simply translates $\varphi$ step by step into a $\SATprobpolysumsm$ instance using the transformation above.
    If the new instance is satisfiable, then it has a model with support $\le n+2$.
\end{proof}

\section{A brief overview of algebraic complexity classes}
\label{sec:appen:alg:classes}

\subsection{\texorpdfstring{$\VP$}{VP}}\label{sec:appen:VP}
Let $X = (X_1,X_2,\dots)$ be an infinite family of indeterminates over some field $F$.

\begin{definition} \label{def:mb:p-fam}
A sequence of polynomials $(f_n) \in F[X]$ is called a \emph{p-family}
if for all $n$,
\begin{enumerate}
\item $f_n \in F[X_1,\dots,X_{p(n)}]$ for some polynomially bounded function $p$ and
\item $\deg f_n \le q(n)$ for some polynomially bounded function $q$. 
\end{enumerate}
\end{definition}

\begin{definition} \label{def:mb:vp}
The class $\VP$ consists of all p-families $(f_n)$ such that 
$L(f_n)$ is polynomially bounded, where $L(f_n)$  denotes circuit complexity of $f_n$ defined 
as the size of smallest arithmetic circuit computing  $f_n$.
\end{definition}

Let $\det_n = \sum_{\pi \in \aS_n} \sgn(\pi) X_{1,\pi(1)} \dots X_{n,\pi(n)}$. 
It is well-known that 
$\det_n$ has polynomial-sized arithmetic circuits \citep{csanky1976fast}. Therefore, $(\det_n) \in \VP$.

In the above example, the indeterminates have two indices instead of one. Of course we could
write $\det_n$ as a polynomial in $X_1,X_2,\dots$ by using a bijection between 
$\Nset^2$ and $\Nset$.
However, we prefer the natural naming of the variables (and will do so with other polynomials).

Let $f \in F[X]$ be a polynomial and $s: X \to F[X]$ be a mapping that maps indeterminates
to polynomials. $s$ can be extended in a unique way to an algebra endomorphism $F[X] \to F[X]$. 
We call $s$ a \emph{substitution}. (Think of the variables replaced by polynomials.)

\begin{definition}
\begin{enumerate}
\item Let $f,g \in F[X]$. $f$ is called a \emph{projection} of $g$ if there is a substitution
$r: X \to X \cup F$ such that $f = r(g)$. We write $f \lep g$ in this case.
(Since $g$ is a polynomial, it only depends on a finite number of indeterminates. 
Therefore, we only need to specify a finite part of $r$.)
\item Let $(f_n)$ and $(g_n)$ be p-families. $(f_n)$ is a \emph{p-projection} of $(g_n)$ if
there is a polynomially bounded function $q: \IN \to \IN$ such that $f_n \lep g_{q(n)}$.
We write $(f_n) \lep (g_n)$.
\end{enumerate}
\end{definition}

Projections are very simple reductions. Therefore, we can also use them to define hardness for
``small'' complexity classes like $\VP$.



\begin{definition} \label{def:vp:compl}
\begin{enumerate}
\item A p-family $(f_n)$ is called \emph{$\VP$-hard} (under p-projections) if
$(g_n) \lep (f_n)$ for all $(g_n) \in \VP$. 
\item It is called \emph{$\VP$-complete} if in addition $(f_n) \in \VP$.
\end{enumerate}
\end{definition}






We do not know whether the determinant (or the iterated matrix multiplication polynomial)
is $\VP$-complete. However, there are generic problems that are $\VP$-complete.
But also more natural complete problems are known, see 
\cite{durand2014homomorphism}.

\subsection{\texorpdfstring{$\VNP$}{VNP}}\label{sec:appen:VNP}


\begin{definition}
A p-family $(f_n)$ is in $\VNP$, if there 
are polynomially bounded functions $p$ and $q$ and a sequence $(g_n) \in \VP$ of polynomials
$g_n \in F[X_1,\dots,X_{p(n)},Y_1,\dots,Y_{q(n)}]$ such that
\[
   f_n = \sum_{e \in \{0,1\}^{q(n)}} g_n(X_1,\dots,X_{p(n)},e_1,\dots,e_{q(n)}).
\]
\item A family of polynomials $f_n$ is in $\VNP_e$
if in the definition of $\VNP$, the family $(g_n)$ is in $\VP_e$.
\end{definition}

You can think of the $X$-variables representing the input and the $Y$-variables
the witness. With this interpretation, $\VNP$ is more like $\sharpP$.

The hardness and completeness for $\VNP$ are defined as in Definition~\ref{def:vp:compl}.
The permanent polynomial 
\[
   \per_n = \sum_{\sigma \in \aS_n} X_{1,\sigma(1)} \cdots X_{n,\sigma(n)}
\]
is complete for $\VNP$. 


\end{document}